\documentclass{article}

    \usepackage[final]{neurips_2024}

\usepackage[utf8]{inputenc} % allow utf-8 input
\usepackage[T1]{fontenc}    % use 8-bit T1 fonts
\usepackage{hyperref}       % hyperlinks
\usepackage{url}            % simple URL typesetting
\usepackage{booktabs}       % professional-quality tables
\usepackage{amsfonts}       % blackboard math symbols
\usepackage{nicefrac}       % compact symbols for 1/2, etc.
\usepackage{microtype}      % microtypography
\usepackage{xcolor}         % colors
\usepackage{amsmath}
\usepackage{amssymb}
\usepackage{mathtools}
\usepackage{amsthm}
\usepackage[style=numeric-comp, sorting=none]{biblatex}
\usepackage{graphicx,subfigure}
\usepackage{algorithm}
\usepackage{algorithmic}
\usepackage{cleveref}
\theoremstyle{plain}
\newtheorem{theorem}{Theorem}[section]

\newtheorem{lemma}[theorem]{Lemma}

\theoremstyle{definition}
\newtheorem{definition}[theorem]{Definition}
\newtheorem{assumption}[theorem]{Assumption}
\theoremstyle{remark}
\newtheorem{remark}[theorem]{Remark}
\usepackage{multirow}
\usepackage{multicol}
\usepackage{framed}
\usepackage{tabstackengine}
\usepackage{makecell}
\usepackage{colortbl}

\title{RFLPA: A Robust Federated Learning Framework against Poisoning Attacks with Secure Aggregation}

\author{%
  Peihua Mai \\
  \And
  Ran Yan \\
  National University of Singapore \\
  \And
  Yan Pang \thanks{Correspondence to \texttt{bizpyj@nus.edu.sg}}
}

\begin{document}

\maketitle

\begin{abstract}
Federated learning (FL) allows multiple devices to train a model collaboratively without sharing their data. Despite its benefits, FL is vulnerable to privacy leakage and poisoning attacks. To address the privacy concern, secure aggregation (SecAgg) is often used to obtain the aggregation of gradients on sever without inspecting individual user updates. Unfortunately, existing defense strategies against poisoning attacks rely on the analysis of local updates in plaintext, making them incompatible with SecAgg. To reconcile the conflicts, we propose a robust federated learning framework against poisoning attacks (RFLPA) based on SecAgg protocol. Our framework computes the cosine similarity between local updates and server updates to conduct robust aggregation. Furthermore, we leverage verifiable packed Shamir secret sharing to achieve reduced communication cost of $O(M+N)$ per user, and design a novel dot product aggregation algorithm to resolve the issue of increased information leakage. Our experimental results show that RFLPA significantly reduces communication and computation overhead by over $75\%$ compared to the state-of-the-art secret sharing method, BREA, while maintaining competitive accuracy.
\end{abstract}

\section{Introduction}
\label{sec:intro}

Federated learning (FL) \cite{mcmahan2016federated, ye2023feddisco, ye2024openfedllm, wang2020tackling, yefake} is a promising machine learning technique that has been gaining attention in recent years. It enables numerous devices to collaborate on building a machine learning model without sharing their data with each other. Compared with traditional centralized machine learning, FL preserves the data privacy by ensuring that sensitive data remain on local devices. 

Despite its benefits, FL still has two key concerns to be addressed. Firstly, there is a threat of privacy leakage from local update. Recent works have demonstrated that the individual updates could reveal sensitive information, such as properties of the training data \cite{melis2019exploiting, fredrikson2015model}, or even allows the server to reconstruct the training data \cite{zhu2019deep, chai2020secure}. The second issue is that FL is vulnerable to poisoning attacks. Indeed, malicious users could send manipulated updates to corrupt the global model at their will \cite{bagdasaryan2020backdoor}. The poisoning attacks may degrade the performance of the model, in the case of \textit{untargeted attacks}, or bias the model's prediction towards a specific target labels, in the case of \textit{targeted attacks} \cite{huang2011adversarial}.

Secure aggregation (SecAgg) has become a potential solution to address the privacy concern. Under SecAgg protocol, the server could obtain the sum of gradients without inspecting individual user updates \cite{bonawitz2017practical,bell2020secure}. However, this protocol poses a significant challenge in resisting poisoning attacks in FL. Most defense strategies \cite{blanchard2017machine, cao2021fltrust} require the server to access local updates to detect the attackers, which increases the risk of privacy leakage. The contradiction makes it difficult to develop a FL framework that simultaneously resolves the privacy and robustness concerns.

To our best knowledge, BREA is the state-of-the-art FL framework that defends against poisoning attacks using secret sharing-based SecAgg protocol \cite{so2020byzantine}. Based on verifiable secret sharing, their framework leverages pairwise distances to remove outliers. However, their work is limited by the scaling concerns arising from computation and communication complexity. For a model with dimension $M$ and $N$ selected clients, the framework incurs $O(MN+N)$ communication per user, and $O((N^2+MN)\log^2 N \log \log N)$ computation for the server due to the costly aggregation rule. Furthermore, BREA makes unrealistic assumptions that the users could establish direct communication channels with other mobile devices. 

To address the above challenge, we propose a robust federated learning framework against poisoning attacks (RFLPA) based on SecAgg protocol. We leverage verifiable packed Shamir secret sharing to compute the cosine similarity and aggregate gradients in a secure manner with reduced communication cost of $O(M+N)$ per user. To resolve the increased information leakage from packed secret sharing, we design a dot product aggregation protocol that only reveals a single value of the dot product to the server. Our framework requires the server to store a small and clean root dataset as the benchmark. Each user relies on the server to communicate the secret with each other, and utilizes encryption and signature techniques to ensure the secrecy and integrity of messages. The implementation is available at https://github.com/NusIoraPrivacy/RFLPA.

Our main contributions involves the following:

(1) We propose a federated learning framework that overcomes privacy and robustness issues with reduced communication cost, especially for high-dimensional models. The convergence analysis and empirical results show that our framework maintains competitive accuracy while reducing communication cost significantly.

(2) To protect the privacy of local gradients, we propose a novel dot product aggregation protocol. Directly using packed Shamir secret sharing for dot product calculation can result in information leakage. Our dot product aggregation algorithm addresses this issue by ensuring that the server only learns the single value of the dot product and not other information about the local updates. Furthermore, the proposed protocol enables degree reduction by converting the degree-2d partial dot product shares into degree-d final product shares.

(3) Our framework guarantees the secrecy and integrity of secret shares for a server-mediated network model using encryption and signature techniques.

\section{Literature Review}
% In this section, we summarize the related works in defense strategy against poisoning attackers and robust privacy-preserving FL.

\subsection{Defense against Poisoning Attacks.} 
Various robust aggregation rules have been proposed to defend against poisoning attacks. KRUM selects the benign updates based on the pairwise Euclidean distances between the gradients \cite{blanchard2017machine}. Yin et al. \cite{yin2018byzantine} proposes two robust coordinate-wise aggregation rules that computes the median and trimmed mean at each dimension, respectively. Bulyan \cite{guerraoui2018hidden} selects a set of gradients using Byzantine–resilient algorithm such as KRUM, and then aggregates the updates with trimmed mean. RSA \cite{shi2022challenges} adds a regularization term to the objective function such that the local models are encouraged to be similar to the global model. In FLTrust \cite{cao2021fltrust}, the server maintains a model on its clean root dataset and computes the cosine similarity to detect the malicious users. The aforementioned defense strategies analyze the individual gradients in plaintext, and thus are susceptible to privacy leakage.

\subsection{Robust Privacy-Preserving FL.}
To enhance privacy and resist poisoning attacks, several frameworks have integrated homomorphic encryption (HE) with existing defense techniques. Based on Paillier cryptosystem, PEFL \cite{liu2021privacy} calculates the Pearson correlation coefficient between coordinate-wise medians and local gradients to detect malicious users. PBFL \cite{miao2022privacy} uses cosine similarity to identify poisonous gradients and adopted fully homomorphic encryption (FHE) to ensure security. ShieldFL \cite{ma2022shieldfl} computes cosine similarity between encrypted gradients with poisonous baseline for Byzantine-tolerance aggregation. The above approaches inherit the costly computation overhead of HE. Furthermore, they rely two non-colluding parties to perform secure computation and thus might be vulnerable to privacy leakage. Secure Multi-party Computation (SMC) is an alternative to address the privacy concern. To the best of our knowledge, BREA \cite{so2020byzantine} is the first work that developed Byzantine robust FL framework using verifiable Shamir secret sharing. However, their method suffers high communication complexity of SMC and high computation complexity of KRUM aggregation protocol. Refer to Appendix \ref{app:compareframe} for a comprehensive comparison among existing protocols.

This paper explores the integration of SMC with defense strategy against poisoning attacks. We develop a framework that reduces communication cost, employs a more efficient aggregation rule and guarantees the security for a server-mediated model. 

\section{Problem Formulation and Background}
\subsection{Problem Statement}
We assume that the server trains a model $\mathbf{w}$ with $N$ mobile clients in a federated learning setting. All parties are assumed to be computationally bounded. Each client holds a local dataset $\{D_i\}_{i\in [N]}$, and the server owns a small, clean root dataset $D_0$. The objective is to optimize the expected risk function:
\begin{equation}
    F(\mathbf{w}) = \min_{\mathbf{w}} \mathbb{E}_{D\sim \chi} L(D, \mathbf{w}),
\end{equation}
where $L(D, \mathbf{w})$ is a empirical loss function given dataset $D$. 

In federated learning, the server aggregates local gradients $\mathbf{g}_i^t$ to obtain global gradient $\mathbf{g}^t$ for model update:
\begin{equation}
    \mathbf{g}^t = \sum_{i\in S} \eta_i^t \mathbf{g}_i^t,\ \mathbf{w}^t = \mathbf{w}^{t-1} - \gamma^t \mathbf{g}^t,
\end{equation}
where $\eta_i$ is the weight of client $i$, $\gamma^t$ is the learning rate, and $S$ is the set of selected clients.

% Suppose that the server is curious in that sense that it will attempt to learn information from received messages, and may falsify clients' secret shares transmitted between others. Furthermore, up to $M$ clients could be controlled by probabilistic polynomial time (PPT) adversary that shares malicious information to degrade model performance. 

\subsection{Adversary Model}
We consider two types of users, i.e., honest users and malicious users. The definitions of honest and malicious users are given as follows.
\begin{definition} [Honest Users] 
A user $u$ is honest if and only if $u$ honestly submits its local gradient $g_u$, where $g_u$ is the true gradients trained on its local dataset $D_u$.
\end{definition}
\begin{definition} [Malicious Users] 
A user $u$ is malicious if and only if $u$ is manipulated by an adversary who launches model poisoning attack by submitting poisonous gradients $g_u^*$.
\end{definition}

Server aims to infer users’ information with two types of attacks, i.e., passive inference and active inference attack. In passive inference attack, the server tries to infer users’ sensitive information by the intermediate result it receives from the user or eardrops during communication. In active inference attack, the server would manipulate certain users’ messages to obtain the private values of targeted users.

\subsection{Design Goals}
We aim to design a federated learning system with three goals.
% \begin{itemize}
%     \item \textbf{Privacy}. Under federated learning, users might still be concerned about the information leakage from individual gradients. To protect privacy, the server shouldn't have access to local update of any user. Instead, the server learns only the aggregation weights and global gradients, ensuring that individual user data remains protected.
%     \item \textbf{Robustness}. We aim to design a method resilient to model poisonous attack, meaning that the model accuracy should be within a reasonable range under malicious clients.
%     \item \textbf{Efficiency}. Our framework should maintain computation and communication efficiency even if it's operated on high dimensional vectors.
% \end{itemize}

\textbf{Privacy}. Under federated learning, users might still be concerned about the information leakage from individual gradients. To protect privacy, the server shouldn't have access to local update of any user. Instead, the server learns only the aggregation weights and global gradients, ensuring that individual user data remains protected.

\textbf{Robustness}. We aim to design a method resilient to model poisonous attack, meaning that the model accuracy should be within a reasonable range under malicious clients.

\textbf{Efficiency}. Our framework should maintain computation and communication efficiency even if it is operated on high dimensional vectors.
 % To enhance users' privacy, SecAgg computes the aggregation of gradients without inspecting individual user updates. 

\subsection{Cryptographic Primitives}
In this section we briefly describe cryptographic primitives for our framework. For more details refer to Appendix \ref{app:cryptoprimitive}.

\textbf{Packed Shamir Secret Sharing.} This study uses a generalization of Shamir secret sharing scheme \cite{shamir1979share}, known as "packed secret-sharing" that allows to represent multiple secrets by a single polynomial \cite{franklin1992communication}. A degree-$d$ ($d\geq l-1$) packed Shamir sharing of $\mathbf{s} = (s_1, s_2, ..., s_l)$ stores the $l$ secrets at a polynomial $f(\cdot)$ of degree at most $d$. The secret sharing scheme requires $d+1$ shares for reconstruction, and any $d-l+1$ shares reveals no information of the secret.

\textbf{Key Exchange.} The framework relies on Diffie–Hellman key exchange protocol \cite{diffie2022new} that allows two parties to establish a secret key securely. 

\textbf{Symmetric Encryption.} Symmetric encryption guarantees the secrecy for communication between two parties \cite{delfs2007symmetric}. The encryption and decryption are conducted with the same key shared by both communication partners.

\textbf{Signature Scheme.} To ensure the integrity and authenticity of message, we adopt a UF-CMA secure signature scheme \cite{kaur2012digital, katz2010digital}.

\section{Framework}
\subsection{Overview}
\begin{figure*}[htbp!]
    \centering    \centerline{\includegraphics[width=1\columnwidth]{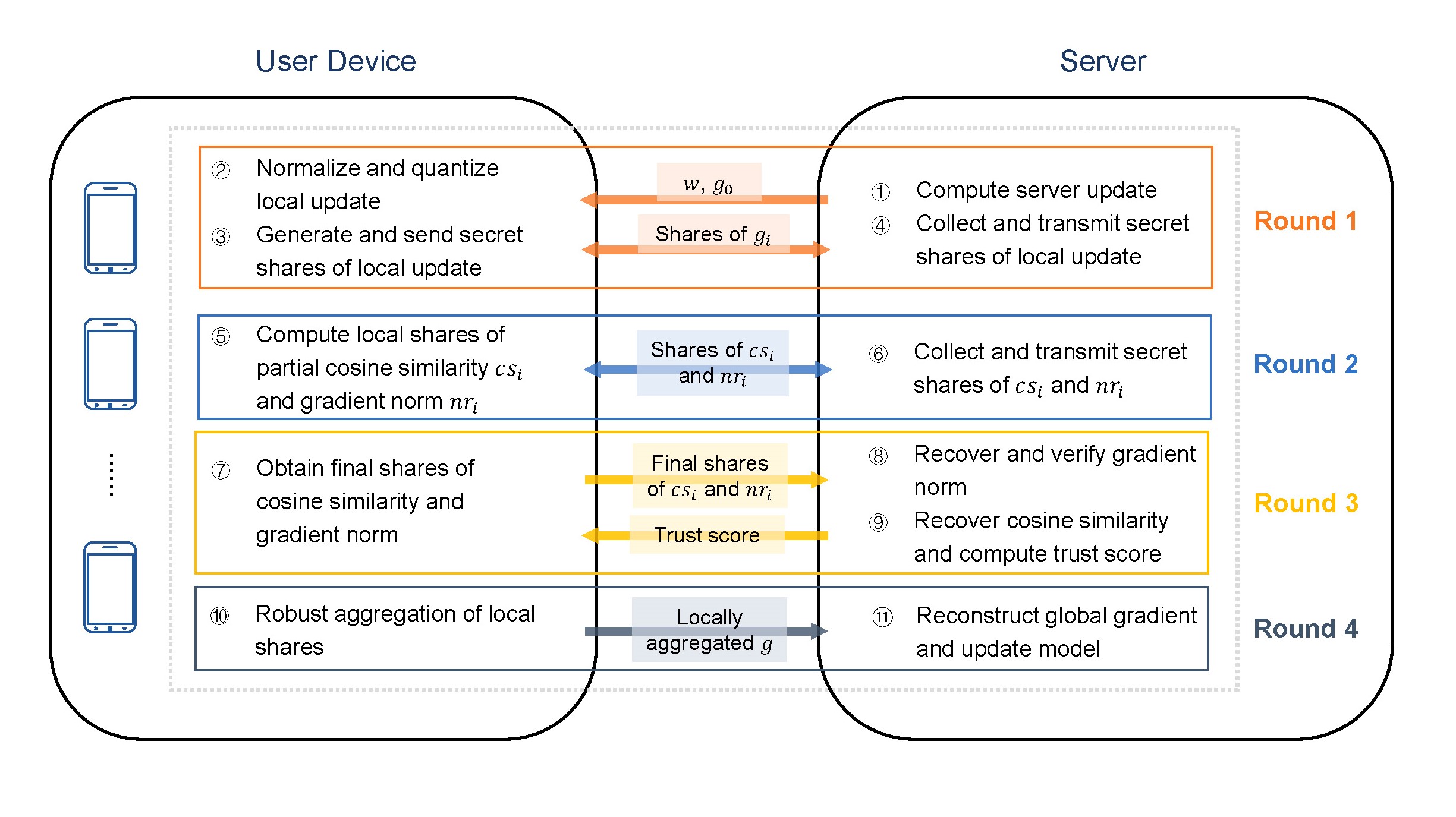}}
    \caption{Overall framework}
    \label{fig:overview}
\end{figure*}
Figure \ref{fig:overview} depicts the overall framework of our robust federated learning algorithm. The algorithm consists of four rounds:

\textbf{Round 1:} each client receives the server update $g_0$, computes their updates normalized by $g_0$, and distributes the secret shares of their updates to other clients.

\textbf{Round 2:} each client computes the local shares of partial dot product for gradient norm and cosine similarity, and conducts secret re-sharing on the local shares.

\textbf{Round 3:} each client obtains final shares of partial dot product for gradient norm and cosine similarity, and transmits the shares to server. Then the server would verify the gradient norm, recover cosine similarity, and compute the trust score for each client.

\textbf{Round 4:} on receiving the trust score from the server, each client conducts robust aggregation on the secret shares locally, and transmits the secret shares of aggregated gradient to the server. The server finally reconstructs the aggregation on the secret shares.

% To detect potentially malicious clients, we calculate the trust score of each user derived from the cosine similarity between user and server update. Our framework leverages verifiable packed Shamir secret sharing to compute update norm, cosine similarity, and gradient aggregation in a secure manner. 

To address increased information leakage caused by packed secret sharing, we design a dot product aggregation protocol to sum up the dot product over sub-groups of elements. 
% The clients rely on server to communicate the secret shares with each other, and adopt encryption and signature techniques to ensure the secrecy and integrity of messages. 
Refer to Appendix \ref{app:algorithm} for the algorithm to perform robust federated learning.

\subsection{Normalization and Quantization}
To limit the impact of attackers, we follow \cite{cao2021fltrust} to normalize each local gradient based on the server model update:
\begin{equation}
    \mathbf{\bar{g}}_i = \frac{\|\mathbf{g}_0\|}{\|\mathbf{g}_i\|} \cdot \mathbf{g}_i,
\end{equation}
where $\mathbf{g}_i$ is the local gradient of the $i$th client, and $\mathbf{g}_0$ is the server gradient obtained from clean root data.

Each client performs local gradient normalization, and the server validates if the updates are truly normalized. The secret sharing scheme operates over finite field $\mathbb{F}_p$ for some large prime number $p$, and thus the user should quantize their normalized update $\mathbf{\bar{g}}_i$. The quantization poses challenge on normalization verification, as $\|\mathbf{\bar{g}}_i\|$ might not be exactly equal to $\|\mathbf{g}_0\|$ after being converted into finite field. 

To address this issue, we define the following rounding function:
\begin{equation}
    Q(x) = \left\{
    \begin{array}{lc}
    \lfloor qx \rfloor/q, & x \geq 0\\
    (\lfloor qx \rfloor+1)/q, & x < 0\\
    \end{array}
    \right.,
\end{equation}
where $\lfloor qx \rfloor$ is the largest integer less than or equal to $qx$.

Therefore, the server could verify that $\|\mathbf{\bar{g}}_i\| \leq \|\mathbf{g}_0\|$, which is ensured by the quantization method.

\subsection{Robust Aggregation Rule}
Consistent with FLTrust\cite{cao2021fltrust}, our framework conducts robust aggregation using the cosine similarity between users' and server's updates. The trust score of user $i$ is:
\begin{equation}
\label{eq:trustscore}
    TS_i = \max\left(0, \frac{\langle\mathbf{g}_i, \mathbf{g}_0\rangle}{\|\mathbf{g}_i\|\|\mathbf{g}_0\|}\right) = \max\left(0, \frac{\langle\mathbf{\bar{g}}_i, \mathbf{g}_0\rangle}{\|\mathbf{g}_0\|^2}\right),
\end{equation}
where we clip the negative cosine similarity to zero to avoid the impact of malicious clients.

The global gradient is then aggregated by:
\begin{equation}
\label{eq:robustagg}
    \mathbf{g} = \frac{1}{\sum_{i=1}^N TS_i} \sum_{i=1}^N TS_i \cdot \mathbf{\bar{g}}_i.
\end{equation}

Finally, we use the gradient to update the global model:
\begin{equation}
    \mathbf{w} \leftarrow \mathbf{w} - \gamma \mathbf{g}.
\end{equation}

Our framework leverages the robust aggregation rule consistent with FLTrust due to its advantages including low computation cost, the absence of a requirement for prior knowledge about number of poisoners, defend against majority number of poisoners, and compatibility with Shamir Secret Sharing. Appendix \ref{app:compaggrule} details the comparison between FLTrust and existing robust aggregation rules.

\subsection{Verifiable Packed Secret Sharing}
The core idea of packed secret sharing is to encode $l$ secrets within a single polynomial. Consequently, the secret shares of local updates generated by each user would reduce from $NM$ to $NM/l$. By selecting $l=O(N)$, the per-user communication cost at secret sharing stage can be decreased to $O(M+N)$. We assume that the prime number $P$ is large enough such that $P>\max\{N\|\mathbf{g}_0\|, \|\mathbf{g}_0\|^2\}$ to avoid overflow.

One issue with secret sharing is that a malicious client may send invalid secret shares, i.e., shares that are not evaluated at the same polynomial function, to break the training process. To address this issue, the framework utilizes the verifiable secret sharing scheme from \cite{kate2010constant}, which generates constant size commitment to improve communication efficiency. We construct the verifiable secret shares for both local gradients and partial dot products described in Section \ref{sec:dotprodagg}. During verifiable packed secret sharing, the user would send the secret shares $\mathbf{s}$, commitment $\mathcal{C}$, and witness $w_l$ to other users. A commitment is a value binding to a polynomial function $\phi(x)$, i.e., the underlying generator of the secret shares, without revealing it. A witness allows others to verify that the secret share $s_l$ is generated at $l$ of the polynomial (see Appendix \ref{app:vpss} for more details). 

\subsection{Dot Product Aggregation}

Directly applying packed secret sharing may increase the risk of information leakage when calculating cosine similarity and gradient norm. In the example provided by Figure \ref{fig:dotprod}, the gradient vectors are created as secret shares by packing $l$ secret into a polynomial function. Following the local similarity computations by each client, the server can reconstruct the element-wise product between the two gradients, which makes it easy to recover the user's gradient $\Tilde{g}_i$ from the reconstructed metric. On the other hand, our proposed protocol ensures that only the single value of dot product is released to the server. Based on this, we introduce a term \textit{partial dot product}, or \textit{partial cosine similarity (norm square)} depending on the input vectors, defined as follows:

\textit{\textbf{Partial dot product} represents the multiple dot products of several subgroups of elements from input vectors rather than a single dot product value.}

Another related concept is \textit{final dot product}, referring to the single value of dot products between two vectors. For example, given two vectors $\boldsymbol{v}_1=(2, -1, 4, 5, 6, 3)$ and $\boldsymbol{v}_2=(1, 2, 0, 3, -2, 1)$, the reconstructed \textit{partial dot product} could be $(0, 15, -9)$ if we pack $2$ elements into a secret share, while the \textit{final dot product} is $6$. If each client directly uploads the shares from local dot product computation, the server would reconstruct a vector of partial cosine similarity (norm square) and thus learn more gradient information.
\label{sec:dotprodagg}
\begin{figure}[htbp!]
    \centering    \centerline{\includegraphics[width=0.7\columnwidth]{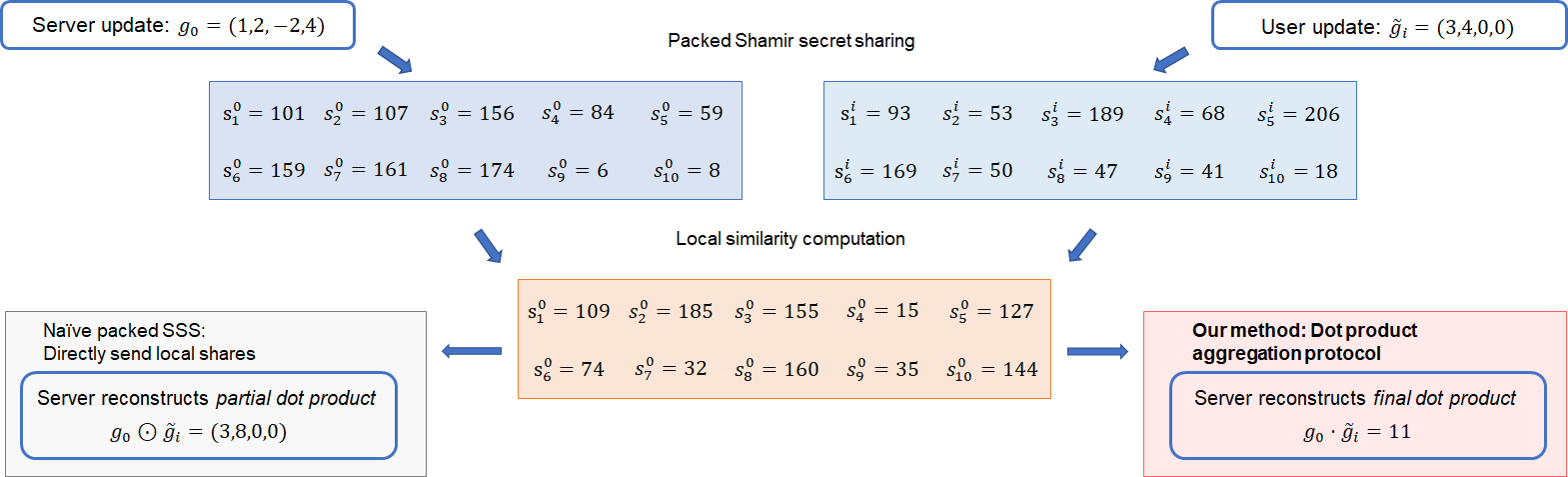}}
    \caption{Cosine similarity computation on packed secret sharing}
    \label{fig:dotprod}
\end{figure}
To ensure that the server only has access to final cosine similarity (norm square), we design a dot product aggregation algorithm based on secret re-sharing that allows the users to sum up the dot products over subgroups. 

Suppose that the user $i$ creates a packed secret sharing $\mathbf{V}^i=\{v^i_{jk}\}_{j\in [N], k\in [\lceil m/l \rceil]}$ of $\mathbf{\Tilde{g}}_i=(g_1^i, g_2^i, ..., g_M^i)$, by packing each $l$ elements into a secret. On receiving the secret shares, each user $i$ can compute the vectors $\mathbf{cs}^i=(cs_1^i, cs_2^i, ..., cs_N^i)$ and $\mathbf{nr}^i=(nr_1^i, nr_2^i, ..., nr_N^i)$:
\begin{equation}
\begin{gathered}
\label{eq:cosnorm}
    cs_j^i = \sum_l v_{il}^j \cdot v_{il}^0,\ nr_j^i = \sum_l v_{il}^j \cdot v_{il}^j,
\end{gathered}
\end{equation}
 where $cs_j^i$ and $nr_j^i$ denotes the $i^{th}$ share of partial cosine similarity and partial gradient norm square for user $j$'s gradient.

The partial cosine similarity (or gradient norm square) could be further aggregated by the procedure below in four steps. 

\textbf{Step 1: Secret resharing of partial dot product.} Each user $i$ could construct the verifiable packed secret shares of $\mathbf{cs}^i$ (or $\mathbf{nr}^i$) by representing $p$ secrets on a polynomial:
\begin{equation}
    \mathbf{S}^i = \left(\begin{array}{ccc}
    s_{11}^i& \dots & s_{1\lceil N/p \rceil}^i \\
    \vdots & \ddots & \vdots \\
    s_{N1}^i & \dots & s_{N\lceil N/p \rceil}^i
    \end{array}
    \right),
\end{equation}
where $s_{jk}^i$ denotes the share sent to user $j$ for the $k^{th}$ group of elements in vector $\mathbf{cs}^i$ (or $\mathbf{nr}^i$). By choosing $p=O(N)$, each user will generate $O(N)$ secret shares.

\textbf{Step 2: Disaggregation on re-combination vector.} After distributing the secret shares, each user $i$ receives a re-combination vector $\mathbf{s}_{ik} = (s_{ik}^1, s_{ik}^2, ..., s_{ik}^N)$ for $k\in [\lceil N/p \rceil]$. Since we pack $l$ elements for the secret shares of partial dot product, this step aims to transform the $\mathbf{s}_{ik}$ into $l$ vectors, with each vector representing one element. For each $j\in [l]$, user $i$ locally computes:
% \begin{equation}
% \label{eq:multi2single}
%     \mathbf{\Tilde{w}}_{jk}^i = \mathbf{s}_{ik} B_{e_j}^{-1} Chop_d B_{e_{j'}}
% \end{equation}
\begin{equation}
\label{eq:multi2single}
    \mathbf{\Tilde{h}}_{jk}^i = \mathbf{s}_{ik} B_{e_j}^{-1} Chop_d,
\end{equation}
where $B_{e_j}$ is an $n$ by $n$ matrix whose $(i,k)$ entry is $(\alpha_k-e_j)^{i-1}$, and $Chop_d$ is an $n$ by $n$ matrix whose $(i,k)$ entry is $1$ if $1\leq i=k\leq d$ and $0$ otherwise. After this operation, the degree-2d partial dot product shares are transformed into degree-d shares.

\textbf{Step 3: Aggregation along packed index.} The new secrets are summed up along $j\in [l]$ at client side:
\begin{equation}
\label{eq:aggsubgroup}
    \mathbf{h}_k^i = \sum_{j=1}^l \mathbf{\Tilde{h}}_{jk}^i.
\end{equation}

\textbf{Step 4: Decoding for final secret shares.} User $i$ can derive the final secret shares $x_k^i$ by recovering from $\mathbf{h}_k^i=(h_{k1}^i, h_{k2}^i, ..., h_{kN}^i)$ using Reed-Solomon decoding. Noted that $\{x_k^i\}_{k\in [\lceil N/p \rceil]}$ becomes a packed secret share of dot products of degree $d$ (see Appendix \ref{app:secretreshare}). Therefore, the server could recover the cosine similarity (or gradient norm square) for all users on receiving the final shares from sufficient users.

\subsection{Secret Sharing over Insecure Channel}
This framework relies on a server-mediated communication channel for the following reasons: (1) it's challenging for mobile clients to establish direct communication with each other and authenticate other devices; (2) a server could act as central coordinator to ensure that all clients have access to the latest model. On the other hand, the secret sharing stage requires to maintain the privacy and integrity of secret shares.

To protect the secrecy of message, we utilize key agreement and symmetric encryption protocol. The clients establish the secret keys with each other through Diffie–Hellman key exchange protocol. During secret sharing, each client $u$ uses the common key $k_{uv}$ to encrypt the message sent to client $v$, and client $v$ could decrypt the cyphertext with the same key.

Another concern is that the server may falsify the messages transmitted between clients. Signature scheme is adopted to prevent the active attack from server. We assume that all clients receive their private signing key and public signing keys of all other clients from a trusted third party. Each client $i$ generates a signature $\sigma_i$ along with the message $m$, and other clients verify the message using client $i$'s public key $d_i^{PK}$.

\section{Theoretical Analysis}
\subsection{Complexity Analysis}
In this section, we analyze the per iteration complexity for $N$ selected clients, and model dimension of $M$, and summarize the complexity in Table \ref{tab:complexity}. Further details of the complexity analysis are available in Appendix \ref{det:complexity analysis}. One important observation is that the communication complexity of our protocol reduces from $O(MN+N)$ to $O(M+N)$. Furthermore, the server-side computation overhead is reduced to $O((M+N) \log^2 N \log \log N)$, benefiting from the efficient aggregation rule and packed secret sharing. It should be noted that while the BERA protocol has similar server communication complexity, it makes an unrealistic assumption that users can share secrets directly with each other, thereby saving the server's overhead.

\begin{table*}[htp]
\caption{Complexity summary of RFLPA and BERA}
\label{tab:complexity}
\begin{center}
\begin{scriptsize}
% \scalebox{0.95}{
\begin{tabular}{lcccc}
\toprule
 &\multicolumn{2}{c}{RFLPA} & \multicolumn{2}{c}{BERA} \\
 & Computation & Communication & Computation & Communication\\
 \midrule
 Server & $O((M+N) \log^2 N \log \log N)$& $O((M+N)N)$ & $O((N^2+MN) \log^2 N \log \log N)$ & $O(MN+N^2)$\\
 User & $O((M+N^2)\log^2 N)$& $O((M+N))$ & $O(MN\log^2 N + MN^2)$ & $O(MN+N)$\\
\bottomrule
\end{tabular}
% }
\end{scriptsize}
\end{center}
\end{table*}

\subsection{Security Analysis}
The security analysis is conducted for Algorithm \ref{alg:robustsecagg}. Given a security parameter $\kappa$, a server $S$, and any subsets of users $\mathcal{U}$, let ${\rm REAL}_{\mathcal{C}}^{\mathcal{U},t,\kappa}$ be a random variable
representing the joint view of parties in $\mathcal{C} \subseteq \mathcal{U} \cup S$ where the threshold is set to $t$, and $\mathcal{U}_i$ be the subset of respondents at round $i$ such that $\mathcal{U} \supseteq \mathcal{U}_1 \supseteq \mathcal{U}_2 \supseteq \mathcal{U}_3 \supseteq \mathcal{U}_4$. We show that the joint view of any group of parties from $\mathcal{C}$ with users less than $t$ can be simulated given the inputs of clients in that group, trust score $\{TS_j\}_{j\in \mathcal{U}_1}$, and global gradient $\mathbf{g}$. In other words, \textit{the server learns no information about clients' input except the global gradient and trust score}.
\begin{theorem} [Security against active server and clients]
\label{thm:security} 
    There exists a PPT simulator {\rm SIM} such that for all $t\leq K-L$, $|\mathcal{C}\backslash \{S\}| < t$, the output of {\rm SIM} is computationally indistinguishable from the output of ${\rm REAL}_{\mathcal{C}}^{\mathcal{U},t,\kappa}$:
    \begin{equation}
        {\rm REAL}_{\mathcal{C}}^{\mathcal{U},t,\kappa}(\mathbf{x}_{\mathcal{U}}) \equiv  {\rm SIM}_{\mathcal{C}}^{\mathcal{U},t,\kappa}(\mathbf{x}_{\mathcal{U}})
    \end{equation}
    where "$\equiv$" represents computationally indistinguishable.
\end{theorem}
% Refer to Appendix \ref{app:security} for the proof.

\subsection{Correctness against Malicious Users}
In this section, we show that our protocol executes correctly under the following attacks of malicious users: (1) sending invalid secret shares; (2) sending shares from incorrect computation of \ref{eq:robustagg}, \ref{eq:cosnorm}, \ref{eq:multi2single}, or \ref{eq:aggsubgroup}. Note that adversaries may also create shares from arbitrary gradients, and we left the discussion of such attack to Section \ref{sec:convergence}.

The first attack arises when the user doesn't generate shares from the same polynomial. Such attempt is prevented by verifiable secret sharing that allows for the verification of share validity by testing \ref{eq:bilineartest}. 

The second attack could be addressed by Reed-Solomon codes. For a degree-$d$ packed Shamir secret sharing with $n$ shares, the Reed-Solomon decoding algorithm could recover the correct result with $E$ errors and $S$ erasures as long as $S+2E+d+1\leq n$.

\subsection{Convergence Analysis}
\label{sec:convergence}

\begin{theorem}
\label{convergence}
Suppose Assumption \ref{converge_ass1}, \ref{converge_ass2},  \ref{converge_ass3} in Appendix \ref{app:converge_ass} hold. For arbitrary number of malicioius clients, the difference between the global model $\mathbf{w}^t$ learnt by our algorithm and the optimal $\mathbf{w}^*$ is bounded. Formally, we have the following inequality with probability at least $1-\delta$:
\begin{equation}
    \|\mathbf{w}^t-\mathbf{w}^*\| \leq (1-\rho)^t \|\mathbf{w}^0-\mathbf{w}^*\| + 12\gamma \Delta_1 + \frac{\gamma\sqrt{d}}{q}
\end{equation}
where $\rho = 1-(\sqrt{1-\mu^2/(4L_g^2)}+24\gamma \Delta_2 + 2\gamma L)$, $\Delta_1 = \nu_1 \sqrt{\frac{2}{|D_0|}} \sqrt{d\log 6+\log (3/\delta)}$, $\Delta_2=\nu_2 \sqrt{\frac{2}{|D_0|}} \sqrt{d\log \frac{18L_2}{\nu_2}+\frac{1}{2}d\log \frac{|D_0|}{d} + \log \left(\frac{6\nu_2^2r\sqrt{D_0}}{\alpha_2\nu_1\delta}\right)}$, $L_2=\max\{L,L_1\}$.
\end{theorem}
\begin{remark}
$\gamma\sqrt{d}/q$ is the noise caused by the quantization process in our algorithm.
\end{remark}

\section{Experiments}
\subsection{Experimental Setup}
\label{sec:expsetup}

\textbf{Dataset.} We use three standard datasets to evaluation the performance of RFLPA: MNIST \cite{lecun1998gradient}, FashionMNIST (F-MNIST) \cite{xiao2017fashion}, and CIFAR-10 \cite{krizhevsky2009learning}. MNIST and F-MNIST are trained on the neural network classification model composed of two convolutional layers and two fully connected layers, while CIFAR-10 is trained and evaluated with a ResNet-9 \cite{he2016deep} model.

s
\textbf{Attacks.} We simulate two types of poisoning attacks: gradient manipulation attack (untargeted) and label flipping attack (targeted). Under gradient manipulation attack, the malicious users generate arbitrary gradients from normal distribution of mean 0 and standard deviation 200. For label flipping attack, the adversaries flip the label from $l$ to $P-l-1$, where $P$ is the number of classes. We consider the proportion of attackers from $0\%$ to $30\%$.

\subsection{Experiment Results}
\subsubsection{Accuracy Evaluation}
We compare our proposed method with several FL frameworks: FedAvg \cite{konevcny2016federated}, Bulyan \cite{guerraoui2018hidden}, Trim-mean \cite{yin2018byzantine}, local differential privacy (LDP) \cite{naseri2020local}, central differential privacy (CDP) \cite{naseri2020local}, and BREA \cite{so2020byzantine}. Refer to Table \ref{tab:compareframe} for the corse-grained comparison between RFLPA and the baselines. Noted that several baselines are not included in the accuracy comparison because: (i) The security of the some schemes relies on the assumption of two non-colluding parties, which is vulnerable in real life. (ii) Some frameworks entail significant computation costs, rendering their implementation in real-life scenarios impractical (see Appendix \ref{app:overheadhe}).
% Bulyan and Trim-mean are robust aggregation rules with no extra protection on privacy, while BREA and RFLPA address both privacy and robustness issues simultaneously. 
Table \ref{tab:acc} summarizes the accuracies for different methods under the two attacks. 

When defense strategy is not implemented, the accuracies of FedAvg decrease as the proportion of attackers increases, with a more significant performance drop observed under gradient manipulation attacks. Benefited from the trust benchmark, our proposed framework, RFLPA, demonstrates more stable performance for up to $30\%$ adversaries compared to other baselines. In the absence of attackers, our method achieves slightly lower accuracies than FedAvg, with an average decrease of $2.84\%$, $4.38\%$and $3.46\%$, respectively, for MNIST, F-MNIST, and CIFAR-10 dataset.

\subsubsection{Overhead Analysis}
To verify the effectiveness of our framework on reducing overhead, we compare the per-iteration communication and computation cost for BREA and RFLPA in Figure \ref{fig:comcost}. For each experiment we set the degree as $0.4N$ and encode $0.1N$ elements within a polynomial. 

The left-most graph presents the overhead with different participating client size using the 1.6M parameter model described in Section \ref{sec:expsetup}. For $M \gg N$, the per-client communication complexity for RFLPA remains stable at around 82.5MB, regardless of user size. Conversely, BREA exhibits linear scalability with the number of participating clients. Our framework reduces the communication cost by over $75\%$ compared with BREA.

The second left graph examines the communication overhead for varying model dimensions with 2,000 participating clients. RFLPA achieves a much lower per-client cost than BREA by leveraging packed secret sharing, leading to a $99.3\%$ reduction in overhead.

The right two figures presents the computation cost under varying client size using a MNIST classifier with 1.6M parameters. Benefiting from the packed VSS, RFLPA reduces both the user and server computation overhead by over $80\%$ compared with BREA. 

\begin{table*}[!htb]
\caption{Accuracy under different proportions of attackers. The values denote the mean $\pm$ standard deviation of the performance.}
\label{tab:acc}
\centering
\begin{tiny}
\scalebox{0.92}{
\begin{tabular}{ll|cccc|cccc}
\toprule
 & & \multicolumn{4}{c|}{Gradient Manipulation} & \multicolumn{4}{c}{Label Flipping} \\
 \multicolumn{2}{l|}{Proportion of Attackers} & No & $10\%$ & $20\%$ & $30\%$ & No & $10\%$ & $20\%$ & $30\%$ \\
 \hline
\multirow{3}{*}{FedAvg} & MNIST & $\mathbf{0.98}$ \tiny$\mathbf{\pm 0.0}$\normalsize  & $0.46$ \tiny$\pm 0.1$\normalsize & $0.40$ \tiny$\pm 0.1$\normalsize & $0.32$ \tiny$\pm 0.0$\normalsize & $\mathbf{0.98}$ \tiny$\mathbf{\pm 0.0}$\normalsize & $\mathbf{0.96}$ \tiny$\mathbf{\pm 0.0}$\normalsize & $0.92$ \tiny$\pm 0.0$\normalsize & $0.82$ \tiny$\pm 0.0$\normalsize \\
 & F-MNIST &  $\mathbf{0.88}$ \tiny$\mathbf{\pm 0.0}$\normalsize  & $0.55$ \tiny$\pm 0.0$\normalsize & $0.51$ \tiny$\pm 0.0$\normalsize & $0.45$ \tiny$\pm 0.1$\normalsize & $\mathbf{0.88}$ \tiny$\mathbf{\pm 0.0}$\normalsize & $0.82$ \tiny$\pm 0.0$\normalsize & $0.73$ \tiny$\pm 0.0$\normalsize & $0.69$ \tiny$\pm 0.0$\normalsize \\
 & CIFAR-10 &  $0.76$ \tiny$\pm 0.3$\normalsize  & $0.14$ \tiny$\pm 0.2$\normalsize & $0.13$ \tiny$\pm 0.8$\normalsize & $0.13$ \tiny$\pm 0.2$\normalsize & $0.76$ \tiny$\pm 0.3$\normalsize & $0.72$ \tiny$\pm 1.1$\normalsize & $0.68$ \tiny$\pm 2.7$ & $0.59$ \tiny$\pm 0.8$ \\
 \hline
\multirow{3}{*}{Bulyan} & MNIST & $0.98$ \tiny$\pm 0.0$\normalsize &  $0.92$ \tiny$\pm 0.0$\normalsize & $0.89$ \tiny$\pm 0.0$\normalsize & $0.87$ \tiny$\pm 0.0$\normalsize & $0.98$ \tiny$\pm 0.0$\normalsize &  $0.91$ \tiny$\pm 0.0$\normalsize & $0.90$ \tiny$\pm 0.0$\normalsize & $0.87$ \tiny$\pm 0.0$\normalsize \\
 & F-MNIST & $0.86$ \tiny$\pm 0.0$\normalsize & $0.73$ \tiny$\pm 0.0$\normalsize & $0.71$ \tiny$\pm 0.1$\normalsize & $0.69$ \tiny$\pm 0.0$\normalsize & $0.86$ \tiny$\pm 0.0$\normalsize & $0.76$ \tiny$\pm 0.0$\normalsize & $0.70$ \tiny$\pm 0.1$\normalsize & $0.68$ \tiny$\pm 0.0$\normalsize \\
  & CIFAR-10 &  $\mathbf{0.77}$ \tiny$\mathbf{\pm 1.0}$\normalsize  & $\mathbf{0.73}$ \tiny$\mathbf{\pm 0.8}$\normalsize & $0.45$ \tiny$\pm 1.2$\normalsize & $0.27$ \tiny$\pm 0.6$\normalsize & $\mathbf{0.77}$ \tiny$\mathbf{\pm 1.0}$\normalsize & $\mathbf{0.72}$ \tiny$\mathbf{\pm 0.2}$\normalsize & $0.62$ \tiny$\pm 1.8$\normalsize & $0.40$ \tiny$\pm 0.9$\normalsize \\
 \hline
\multirow{3}{*}{\shortstack{Trim-\\mean}} & MNIST & $0.98$ \tiny$\pm 0.0$\normalsize & $0.95$ \tiny$\pm 0.0$\normalsize & $0.93$ \tiny$\pm 0.0$\normalsize & $0.91$ \tiny$\pm 0.0$\normalsize & $0.98$ \tiny$\pm 0.0$\normalsize & $0.95$ \tiny$\pm 0.0$\normalsize & $0.92$ \tiny$\pm 0.0$\normalsize & $0.90$ \tiny$\pm 0.0$\normalsize \\
 & F-MNIST &  $0.86$ \tiny$\pm 0.0$\normalsize & $0.81$ \tiny$\pm 0.0$\normalsize & $0.74$ \tiny$\pm 0.0$\normalsize & $0.71$ \tiny$\pm 0.0$\normalsize & $0.86$ \tiny$\pm 0.0$\normalsize & $0.78$ \tiny$\pm 0.0$\normalsize & $0.74$ \tiny$\pm 0.0$\normalsize & $0.73$ \tiny$\pm 0.0$\normalsize \\
 & CIFAR-10 &  $0.76$ \tiny$\pm 1.0$\normalsize  & $0.57$ \tiny$\pm 2.1$\normalsize & $0.51$ \tiny$\pm 1.1$\normalsize & $0.47$ \tiny$\pm 2.2$\normalsize & $0.76$ \tiny$\pm 1.0$\normalsize & $0.71$ \tiny$\pm 1.3$\normalsize & $0.68$ \tiny$\pm 0.7$\normalsize & $0.56$ \tiny$\pm 1.1$\normalsize \\
 \hline
 \multirow{3}{*}{LDP} & MNIST & $0.87$ \tiny$\pm 0.1$\normalsize & $0.13$ \tiny$\pm 0.0$\normalsize & $0.10$ \tiny$\pm 0.0$\normalsize & $0.10$ \tiny$\pm 0.0$\normalsize & $0.87$ \tiny$\pm 0.1$\normalsize & $0.87$ \tiny$\pm 0.3$\normalsize & $0.83$ \tiny$\pm 1.2$\normalsize & $0.77$ \tiny$\pm 2.1$\normalsize \\
 & F-MNIST &  $0.74$ \tiny$\pm 0.1$\normalsize & $0.59$ \tiny$\pm 0.4$\normalsize & $0.53$ \tiny$\pm 1.2$\normalsize & $0.12$ \tiny$\pm 0.0$\normalsize & $0.74$ \tiny$\pm 0.1$\normalsize & $0.63$ \tiny$\pm 0.5$\normalsize & $0.62$ \tiny$\pm 0.2$\normalsize & $0.59$ \tiny$\pm 1.2$\normalsize \\
 & CIFAR-10 & $0.14$ \tiny$\pm 0.2$\normalsize  &  $0.14$ \tiny$\pm 0.2$\normalsize & $0.12$ \tiny$\pm 0.3$\normalsize & $0.12$ \tiny$\pm 0.1$\normalsize & $0.14$ \tiny$\pm 0.2$\normalsize & $0.14$ \tiny$\pm 0.2$\normalsize & $0.14$ \tiny$\pm 0.3$\normalsize & $0.13$ \tiny$\pm 0.1$\normalsize \\
 \hline
  \multirow{3}{*}{CDP} & MNIST & $0.96$ \tiny$\pm 0.0$\normalsize & $0.96$ \tiny$\pm 0.0$\normalsize & $0.95$ \tiny$\pm 0.0$\normalsize & $0.94$ \tiny$\pm 0.0$\normalsize & $0.96$ \tiny$\pm 0.0$\normalsize & $0.96$ \tiny$\pm 0.0$\normalsize & $0.95$ \tiny$\pm 0.3$\normalsize & $0.91$ \tiny$\pm 0.2$\normalsize \\
 & F-MNIST &  $0.83$ \tiny$\pm 0.1$\normalsize & $0.51$ \tiny$\pm 0.1$\normalsize & $0.41$ \tiny$\pm 0.0$\normalsize & $0.34$ \tiny$\pm 0.1$\normalsize & $0.83$ \tiny$\pm 0.1$\normalsize & $0.81$ \tiny$\pm 0.5$\normalsize  &  $0.79$ \tiny$\pm 0.0$\normalsize & $0.78$ \tiny$\pm 0.7$\normalsize  \\
 & CIFAR-10 &  $0.71$ \tiny$\pm 1.2$\normalsize  & $0.12$ \tiny$\pm 0.5$\normalsize & $0.12$ \tiny$\pm 0.3$\normalsize & $0.12$ \tiny$\pm 0.3$\normalsize & $0.71$ \tiny$\pm 1.2$\normalsize & $0.68$ \tiny$\pm 0.7$\normalsize & $0.66$ \tiny$\pm 1.5$\normalsize & $0.63$ \tiny$\pm 1.3$\normalsize \\
 \hline
 \multirow{3}{*}{BREA} & MNIST & $0.94$ \tiny$\pm 0.0$\normalsize &  $0.93$ \tiny$\pm 0.0$\normalsize & $0.93$ \tiny$\pm 0.0$\normalsize & $0.93$ \tiny$\pm 0.0$\normalsize & $0.94$ \tiny$\pm 0.0$\normalsize &  $0.94$ \tiny$\pm 0.0$\normalsize & $0.93$ \tiny$\pm 0.0$\normalsize & $0.93$ \tiny$\pm 0.0$\normalsize \\
 & F-MNIST & $0.84$ \tiny$\pm 0.0$\normalsize & $0.83$ \tiny$\pm 0.0$\normalsize & $0.82$ \tiny$\pm 0.0$\normalsize & $0.81$ \tiny$\pm 0.0$\normalsize & $0.84$ \tiny$\pm 0.0$\normalsize & $0.84$ \tiny$\pm 0.0$\normalsize & $0.82$ \tiny$\pm 0.0$\normalsize & $0.81$ \tiny$\pm 0.0$\normalsize \\
 & CIFAR-10 &  $0.70$ \tiny$\pm 1.0$\normalsize  & $0.69$ \tiny$\pm 1.1$\normalsize & $0.68$ \tiny$\pm 1.9$\normalsize & $0.68$ \tiny$\pm 0.7$\normalsize & $0.70$ \tiny$\pm 1.0$  & $0.70$ \tiny$\pm 2.2$ & $0.67$ \tiny$\pm 0.9$\normalsize & $0.65$ \tiny$\pm 2.7$\normalsize \\
 \hline
  \rowcolor[gray]{.8} & MNIST & $0.96$ \tiny$\pm 0.0$\normalsize  & $\mathbf{0.96}$ \tiny$\mathbf{\pm 0.0}$\normalsize & $\mathbf{0.95}$ \tiny$\mathbf{\pm 0.0}$\normalsize & $\mathbf{0.95}$ \tiny$\mathbf{\pm 0.0}$\normalsize & $0.96$ \tiny$\pm 0.0$\normalsize  & $0.96$ \tiny$\pm 0.0$\normalsize  & $\mathbf{0.95}$ \tiny$\mathbf{\pm 0.0}$\normalsize & $\mathbf{0.95}$ \tiny$\mathbf{\pm 0.0}$\normalsize \\
 \rowcolor[gray]{.8} & F-MNIST & $0.84$ \tiny$\pm 0.0$\normalsize  & $\mathbf{0.84}$ \tiny$\mathbf{\pm 0.0}$\normalsize  & $\mathbf{0.83}$ \tiny$\mathbf{\pm 0.0}$\normalsize & $\mathbf{0.82}$ \tiny$\mathbf{\pm 0.0}$\normalsize & $0.84$ \tiny$\pm 0.0$  & $\mathbf{0.83}$ \tiny$\mathbf{\pm 0.0}$ & $\mathbf{0.83}$ \tiny$\mathbf{\pm 0.0}$ & $\mathbf{0.82}$ \tiny$\mathbf{\pm 0.0}$ \\
\rowcolor[gray]{.8} \multirow{-3}{*}{RFLPA} & CIFAR-10 & $0.74$ \tiny$\pm 2.3$\normalsize  & $0.70$ \tiny$\pm 1.8$\normalsize  & $\mathbf{0.70}$ \tiny$\mathbf{\pm 1.9}$\normalsize & $\mathbf{0.69}$ \tiny$\mathbf{\pm 1.8}$\normalsize & $0.74$ \tiny$\pm 2.3$\normalsize & $0.71$ \tiny$\pm 1.7$\normalsize & $\mathbf{0.70}$ \tiny$\mathbf{\pm 1.6}$\normalsize & $\mathbf{0.69}$ \tiny$\mathbf{\pm 0.8}$\normalsize \\
 \bottomrule
\end{tabular}
}
\end{tiny}
\end{table*}

\begin{figure}[ht]
\begin{center}
  \centering    \centerline{\includegraphics[width=1\columnwidth]{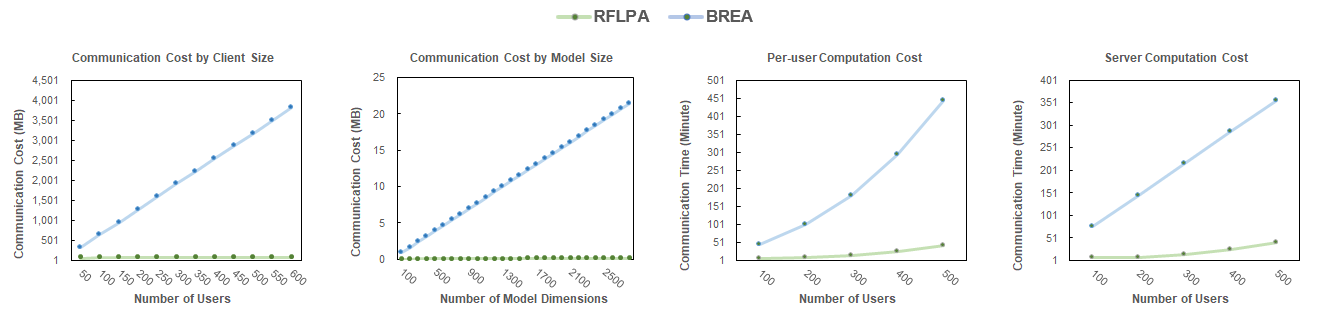}}
\caption{Per-iteration communication (left two) and computation cost (right two).}
\label{fig:comcost}
\end{center}
\end{figure}

\subsubsection{Other studies}
For other studies, we analyze the impact of iterations on accuracy (see Appendix \ref{app:acciter}), evaluate our protocol against additional attacks (see Appendix \ref{app:moreattack}), conduct further overhead analysis (see \ref{app:overhead}), and examine the performance under non-iid setting (see Appendix \ref{app:noniid}).

\section{Conclusion}
This paper proposes RFLPA, a robust privacy-preserving FL framework with SecAgg. Our framework leverages verifiable packed Shamir secret sharing to compute the cosine similarity between user and server update and conduct robust aggregation. We design a secret re-sharing algorithm to address the increased information leakage concern, and utilize encryption and signature techniques to ensure the security over server-mediated channel. Our approach achieves the reduced per-user communication overhead of $O(M+N)$. The empirical study demonstrates that: (1) RFLPA achieves competitive accuracies for up to $30\%$ poisoning adversaries compared with state-of-the-art defense methods. (2) The communication cost and computation cost for RFLPA is significantly lower than BERA by over $75\%$ under the same FL settings. 

\section{Discussion and Future Work}
\label{app:limitation}

\textbf{Collection of server data.} One important assumption is that the server is required to collect a small, clean root dataset. Such collection is affordable for most organizations as the required dataset is of small size, e.g., 200 samples. According to theoretical analysis, the convergence is guaranteed when the root dataset is representative of the overall training data. Empirical evidence presented in \cite{cao2021fltrust} suggests that the performance of the global model is robust even when the root dataset diverges slightly from the overall training data distribution. Furthermore, Appendix \ref{app:cleanremedy} proposes several alternative robust aggregation modules, such as KRUM and comparison with global model, to circumvent the assumption.

% \vspace{0.75em}

\noindent \textbf{Compatibility with other defense strategies.} RFLPA adopts a robust aggregation rule that computes the cosine similarity with server update. The framework can be easily generalized to distance-based method such as KRUM or multi-KRUM by substituting the robust aggregation module. However, extending the framework to rank-based defense methods may be more challenging. Existing SMC techniques for rank-based statistics requires $\log M$ rounds of communication, where $M$ is the range of input values \cite{aggarwal2004secure}. We leave the problem of communication-efficient rank-based robust FL to future work.

% \vspace{0.75em}

\noindent \textbf{Differential privacy guarantee.} Differential privacy (DP) \cite{nguyen2016collecting, berlioz2015applying} provides formal privacy guarantees to prevent information leakage. The combination of SMC and DP, also known as Distributed DP \cite{kairouz2021distributed}, reduces the magnitude of noise added by each user compared with pure local DP. However, adopting DP in the privacy-preserving robust FL framework is non-trivial, especially when bounding the privacy leakage of robustness metrics such as cosine similarity may sacrifice utility. We leave the problem of incorporating DP into the privacy-preserving robust FL framework to future work.

\printbibliography

@article{mcmahan2016federated,
  title={Federated learning of deep networks using model averaging},
  author={McMahan, H Brendan and Moore, Eider and Ramage, Daniel and y Arcas, Blaise Ag{\"u}era},
  journal={arXiv preprint arXiv:1602.05629},
  volume={2},
  year={2016},
  publisher={Technical report}
}

@inproceedings{melis2019exploiting,
  title={Exploiting unintended feature leakage in collaborative learning},
  author={Melis, Luca and Song, Congzheng and De Cristofaro, Emiliano and Shmatikov, Vitaly},
  booktitle={2019 IEEE symposium on security and privacy (SP)},
  pages={691--706},
  year={2019},
  organization={IEEE}
}

@inproceedings{fredrikson2015model,
  title={Model inversion attacks that exploit confidence information and basic countermeasures},
  author={Fredrikson, Matt and Jha, Somesh and Ristenpart, Thomas},
  booktitle={Proceedings of the 22nd ACM SIGSAC conference on computer and communications security},
  pages={1322--1333},
  year={2015}
}

@article{zhu2019deep,
  title={Deep leakage from gradients},
  author={Zhu, Ligeng and Liu, Zhijian and Han, Song},
  journal={Advances in neural information processing systems},
  volume={32},
  year={2019}
}

@article{chai2020secure,
  title={Secure federated matrix factorization},
  author={Chai, Di and Wang, Leye and Chen, Kai and Yang, Qiang},
  journal={IEEE Intelligent Systems},
  volume={36},
  number={5},
  pages={11--20},
  year={2020},
  publisher={IEEE}
}

@inproceedings{bagdasaryan2020backdoor,
  title={How to backdoor federated learning},
  author={Bagdasaryan, Eugene and Veit, Andreas and Hua, Yiqing and Estrin, Deborah and Shmatikov, Vitaly},
  booktitle={International Conference on Artificial Intelligence and Statistics},
  pages={2938--2948},
  year={2020},
  organization={PMLR}
}

@inproceedings{huang2011adversarial,
  title={Adversarial machine learning},
  author={Huang, Ling and Joseph, Anthony D and Nelson, Blaine and Rubinstein, Benjamin IP and Tygar, J Doug},
  booktitle={Proceedings of the 4th ACM workshop on Security and artificial intelligence},
  pages={43--58},
  year={2011}
}

@inproceedings{bonawitz2017practical,
  title={Practical secure aggregation for privacy-preserving machine learning},
  author={Bonawitz, Keith and Ivanov, Vladimir and Kreuter, Ben and Marcedone, Antonio and McMahan, H Brendan and Patel, Sarvar and Ramage, Daniel and Segal, Aaron and Seth, Karn},
  booktitle={proceedings of the 2017 ACM SIGSAC Conference on Computer and Communications Security},
  pages={1175--1191},
  year={2017}
}

@inproceedings{bell2020secure,
  title={Secure single-server aggregation with (poly) logarithmic overhead},
  author={Bell, James Henry and Bonawitz, Kallista A and Gasc{\'o}n, Adri{\`a} and Lepoint, Tancr{\`e}de and Raykova, Mariana},
  booktitle={Proceedings of the 2020 ACM SIGSAC Conference on Computer and Communications Security},
  pages={1253--1269},
  year={2020}
}

@article{blanchard2017machine,
  title={Machine learning with adversaries: Byzantine tolerant gradient descent},
  author={Blanchard, Peva and El Mhamdi, El Mahdi and Guerraoui, Rachid and Stainer, Julien},
  journal={Advances in neural information processing systems},
  volume={30},
  year={2017}
}

@InProceedings{cao2021fltrust,
  title={FLTrust: Byzantine-robust Federated Learning via Trust Bootstrapping},
  author={Cao, Xiaoyu and Fang, Minghong and Liu, Jia and Gong, Neil Zhenqiang},
  booktitle={ISOC Network and Distributed System Security Symposium (NDSS)},
  year={2021}
}

@article{so2020byzantine,
  title={Byzantine-resilient secure federated learning},
  author={So, Jinhyun and G{\"u}ler, Ba{\c{s}}ak and Avestimehr, A Salman},
  journal={IEEE Journal on Selected Areas in Communications},
  volume={39},
  number={7},
  pages={2168--2181},
  year={2020},
  publisher={IEEE}
}

@inproceedings{yin2018byzantine,
  title={Byzantine-robust distributed learning: Towards optimal statistical rates},
  author={Yin, Dong and Chen, Yudong and Kannan, Ramchandran and Bartlett, Peter},
  booktitle={International Conference on Machine Learning},
  pages={5650--5659},
  year={2018},
  organization={PMLR}
}

@inproceedings{guerraoui2018hidden,
  title={The hidden vulnerability of distributed learning in byzantium},
  author={Guerraoui, Rachid and Rouault, S{\'e}bastien and others},
  booktitle={International Conference on Machine Learning},
  pages={3521--3530},
  year={2018},
  organization={PMLR}
}

@article{liu2021privacy,
  title={Privacy-enhanced federated learning against poisoning adversaries},
  author={Liu, Xiaoyuan and Li, Hongwei and Xu, Guowen and Chen, Zongqi and Huang, Xiaoming and Lu, Rongxing},
  journal={IEEE Transactions on Information Forensics and Security},
  volume={16},
  pages={4574--4588},
  year={2021},
  publisher={IEEE}
}

@article{miao2022privacy,
  title={Privacy-preserving Byzantine-robust federated learning via blockchain systems},
  author={Miao, Yinbin and Liu, Ziteng and Li, Hongwei and Choo, Kim-Kwang Raymond and Deng, Robert H},
  journal={IEEE Transactions on Information Forensics and Security},
  volume={17},
  pages={2848--2861},
  year={2022},
  publisher={IEEE}
}

@article{ma2022shieldfl,
  title={ShieldFL: Mitigating model poisoning attacks in privacy-preserving federated learning},
  author={Ma, Zhuoran and Ma, Jianfeng and Miao, Yinbin and Li, Yingjiu and Deng, Robert H},
  journal={IEEE Transactions on Information Forensics and Security},
  volume={17},
  pages={1639--1654},
  year={2022},
  publisher={IEEE}
}

@inproceedings{shi2022challenges,
  title={Challenges and approaches for mitigating byzantine attacks in federated learning},
  author={Shi, Junyu and Wan, Wei and Hu, Shengshan and Lu, Jianrong and Zhang, Leo Yu},
  booktitle={2022 IEEE International Conference on Trust, Security and Privacy in Computing and Communications (TrustCom)},
  pages={139--146},
  year={2022},
  organization={IEEE}
}

@misc{alfred1997handbook,
  title={Handbook of applied cryptography},
  author={Alfred, Menezes and Scott, Vanstone and others},
  year={1997},
  publisher={CRC press}
}

@inproceedings{boneh2004short,
  title={Short signatures without random oracles},
  author={Boneh, Dan and Boyen, Xavier},
  booktitle={Advances in Cryptology-EUROCRYPT 2004: International Conference on the Theory and Applications of Cryptographic Techniques, Interlaken, Switzerland, May 2-6, 2004. Proceedings 23},
  pages={56--73},
  year={2004},
  organization={Springer}
}

@inproceedings{franklin1992communication,
  title={Communication complexity of secure computation},
  author={Franklin, Matthew and Yung, Moti},
  booktitle={Proceedings of the twenty-fourth annual ACM symposium on Theory of computing},
  pages={699--710},
  year={1992}
}

@article{shamir1979share,
  title={How to share a secret},
  author={Shamir, Adi},
  journal={Communications of the ACM},
  volume={22},
  number={11},
  pages={612--613},
  year={1979},
  publisher={ACm New York, NY, USA}
}

@incollection{diffie2022new,
  title={New directions in cryptography},
  author={Diffie, Whitfield and Hellman, Martin E},
  booktitle={Democratizing Cryptography: The Work of Whitfield Diffie and Martin Hellman},
  pages={365--390},
  year={2022}
}

@article{delfs2007symmetric,
  title={Symmetric-key encryption},
  author={Delfs, Hans and Knebl, Helmut and Delfs, Hans and Knebl, Helmut},
  journal={Introduction to cryptography: principles and applications},
  pages={11--31},
  year={2007},
  publisher={Springer}
}

@inproceedings{bellare2000authenticated,
  title={Authenticated encryption: Relations among notions and analysis of the generic composition paradigm},
  author={Bellare, Mihir and Namprempre, Chanathip},
  booktitle={Advances in Cryptology—ASIACRYPT 2000: 6th International Conference on the Theory and Application of Cryptology and Information Security Kyoto, Japan, December 3--7, 2000 Proceedings 6},
  pages={531--545},
  year={2000},
  organization={Springer}
}

@inproceedings{kate2010constant,
  title={Constant-size commitments to polynomials and their applications},
  author={Kate, Aniket and Zaverucha, Gregory M and Goldberg, Ian},
  booktitle={Advances in Cryptology-ASIACRYPT 2010: 16th International Conference on the Theory and Application of Cryptology and Information Security, Singapore, December 5-9, 2010. Proceedings 16},
  pages={177--194},
  year={2010},
  organization={Springer}
}

@book{kung1973fast,
  title={Fast evaluation and interpolation},
  author={Kung, Hsiang-Tsung},
  year={1973},
  publisher={Carnegie-Mellon University. Department of Computer Science}
}

@article{gao2003new,
  title={A new algorithm for decoding Reed-Solomon codes},
  author={Gao, Shuhong},
  journal={Communications, information and network security},
  pages={55--68},
  year={2003},
  publisher={Springer}
}

@article{konevcny2016federated,
  title={Federated learning: Strategies for improving communication efficiency},
  author={Kone{\v{c}}n{\`y}, Jakub and McMahan, H Brendan and Yu, Felix X and Richt{\'a}rik, Peter and Suresh, Ananda Theertha and Bacon, Dave},
  journal={arXiv preprint arXiv:1610.05492},
  year={2016}
}

@inproceedings{aggarwal2004secure,
  title={Secure computation of the k th-ranked element},
  author={Aggarwal, Gagan and Mishra, Nina and Pinkas, Benny},
  booktitle={Advances in Cryptology-EUROCRYPT 2004: International Conference on the Theory and Applications of Cryptographic Techniques, Interlaken, Switzerland, May 2-6, 2004. Proceedings 23},
  pages={40--55},
  year={2004},
  organization={Springer}
}

@article{nguyen2016collecting,
  title={Collecting and analyzing data from smart device users with local differential privacy},
  author={Nguy{\^e}n, Th{\^o}ng T and Xiao, Xiaokui and Yang, Yin and Hui, Siu Cheung and Shin, Hyejin and Shin, Junbum},
  journal={arXiv preprint arXiv:1606.05053},
  year={2016}
}

@inproceedings{berlioz2015applying,
  title={Applying differential privacy to matrix factorization},
  author={Berlioz, Arnaud and Friedman, Arik and Kaafar, Mohamed Ali and Boreli, Roksana and Berkovsky, Shlomo},
  booktitle={Proceedings of the 9th ACM Conference on Recommender Systems},
  pages={107--114},
  year={2015}
}

@inproceedings{kairouz2021distributed,
  title={The distributed discrete gaussian mechanism for federated learning with secure aggregation},
  author={Kairouz, Peter and Liu, Ziyu and Steinke, Thomas},
  booktitle={International Conference on Machine Learning},
  pages={5201--5212},
  year={2021},
  organization={PMLR}
}

@inproceedings{mcmahan2017communication,
  title={Communication-efficient learning of deep networks from decentralized data},
  author={McMahan, Brendan and Moore, Eider and Ramage, Daniel and Hampson, Seth and y Arcas, Blaise Aguera},
  booktitle={Artificial intelligence and statistics},
  pages={1273--1282},
  year={2017},
  organization={PMLR}
}

@inproceedings{he2016deep,
  title={Deep residual learning for image recognition},
  author={He, Kaiming and Zhang, Xiangyu and Ren, Shaoqing and Sun, Jian},
  booktitle={Proceedings of the IEEE conference on computer vision and pattern recognition},
  pages={770--778},
  year={2016}
}

@article{naseri2020local,
  title={Local and central differential privacy for robustness and privacy in federated learning},
  author={Naseri, Mohammad and Hayes, Jamie and De Cristofaro, Emiliano},
  journal={arXiv preprint arXiv:2009.03561},
  year={2020}
}

@article{lecun1998gradient,
  title={Gradient-based learning applied to document recognition},
  author={LeCun, Yann and Bottou, L{\'e}on and Bengio, Yoshua and Haffner, Patrick},
  journal={Proceedings of the IEEE},
  volume={86},
  number={11},
  pages={2278--2324},
  year={1998},
  publisher={Ieee}
}

@article{xiao2017fashion,
  title={Fashion-mnist: a novel image dataset for benchmarking machine learning algorithms},
  author={Xiao, Han and Rasul, Kashif and Vollgraf, Roland},
  journal={arXiv preprint arXiv:1708.07747},
  year={2017}
}

@article{krizhevsky2009learning,
  title={Learning multiple layers of features from tiny images},
  author={Krizhevsky, Alex and others},
  year={2009}
}

@inproceedings{kaur2012digital,
  title={Digital signature},
  author={Kaur, Ravneet and Kaur, Amandeep},
  booktitle={2012 International Conference on Computing Sciences},
  pages={295--301},
  year={2012},
  organization={IEEE}
}

@book{katz2010digital,
  title={Digital signatures},
  author={Katz, Jonathan},
  volume={1},
  year={2010},
  publisher={Springer}
}

@article{pillutla2022robust,
  title={Robust aggregation for federated learning},
  author={Pillutla, Krishna and Kakade, Sham M and Harchaoui, Zaid},
  journal={IEEE Transactions on Signal Processing},
  volume={70},
  pages={1142--1154},
  year={2022},
  publisher={IEEE}
}

@inproceedings{hao2021efficient,
  title={Efficient, private and robust federated learning},
  author={Hao, Meng and Li, Hongwei and Xu, Guowen and Chen, Hanxiao and Zhang, Tianwei},
  booktitle={Proceedings of the 37th Annual Computer Security Applications Conference},
  pages={45--60},
  year={2021}
}

@article{luo2021no,
  title={No fear of heterogeneity: Classifier calibration for federated learning with non-iid data},
  author={Luo, Mi and Chen, Fei and Hu, Dapeng and Zhang, Yifan and Liang, Jian and Feng, Jiashi},
  journal={Advances in Neural Information Processing Systems},
  volume={34},
  pages={5972--5984},
  year={2021}
}

@inproceedings{fang2020local,
  title={Local model poisoning attacks to $\{$Byzantine-Robust$\}$ federated learning},
  author={Fang, Minghong and Cao, Xiaoyu and Jia, Jinyuan and Gong, Neil},
  booktitle={29th USENIX security symposium (USENIX Security 20)},
  pages={1605--1622},
  year={2020}
}

@article{gu2019badnets,
  title={Badnets: Evaluating backdooring attacks on deep neural networks},
  author={Gu, Tianyu and Liu, Kang and Dolan-Gavitt, Brendan and Garg, Siddharth},
  journal={IEEE Access},
  volume={7},
  pages={47230--47244},
  year={2019},
  publisher={IEEE}
}

@inproceedings{yaldiz2023secure,
  title={Secure Federated Learning against Model Poisoning Attacks via Client Filtering},
  author={Yaldiz, Duygu Nur and Zhang, Tuo and Avestimehr, Salman},
  booktitle={ICLR 2023 Workshop on Backdoor Attacks and Defenses in Machine Learning}
}

@inproceedings{ye2023feddisco,
  title={Feddisco: Federated learning with discrepancy-aware collaboration},
  author={Ye, Rui and Xu, Mingkai and Wang, Jianyu and Xu, Chenxin and Chen, Siheng and Wang, Yanfeng},
  booktitle={International Conference on Machine Learning},
  pages={39879--39902},
  year={2023},
  organization={PMLR}
}

@article{wang2020tackling,
  title={Tackling the objective inconsistency problem in heterogeneous federated optimization},
  author={Wang, Jianyu and Liu, Qinghua and Liang, Hao and Joshi, Gauri and Poor, H Vincent},
  journal={Advances in neural information processing systems},
  volume={33},
  pages={7611--7623},
  year={2020}
}

@inproceedings{ye2024openfedllm,
  title={Openfedllm: Training large language models on decentralized private data via federated learning},
  author={Ye, Rui and Wang, Wenhao and Chai, Jingyi and Li, Dihan and Li, Zexi and Xu, Yinda and Du, Yaxin and Wang, Yanfeng and Chen, Siheng},
  booktitle={Proceedings of the 30th ACM SIGKDD Conference on Knowledge Discovery and Data Mining},
  pages={6137--6147},
  year={2024}
}

@inproceedings{yefake,
  title={Fake It Till Make It: Federated Learning with Consensus-Oriented Generation},
  author={Ye, Rui and Du, Yaxin and Ni, Zhenyang and Wang, Yanfeng and Chen, Siheng},
  booktitle={The Twelfth International Conference on Learning Representations}
}

@inproceedings{lycklama2023rofl,
  title={Rofl: Robustness of secure federated learning},
  author={Lycklama, Hidde and Burkhalter, Lukas and Viand, Alexander and K{\"u}chler, Nicolas and Hithnawi, Anwar},
  booktitle={2023 IEEE Symposium on Security and Privacy (SP)},
  pages={453--476},
  year={2023},
  organization={IEEE}
}

@inproceedings{rathee2023elsa,
  title={Elsa: Secure aggregation for federated learning with malicious actors},
  author={Rathee, Mayank and Shen, Conghao and Wagh, Sameer and Popa, Raluca Ada},
  booktitle={2023 IEEE Symposium on Security and Privacy (SP)},
  pages={1961--1979},
  year={2023},
  organization={IEEE}
}

@article{bentivogli2009fifth,
  title={The Fifth PASCAL Recognizing Textual Entailment Challenge.},
  author={Bentivogli, Luisa and Clark, Peter and Dagan, Ido and Giampiccolo, Danilo},
  journal={TAC},
  volume={7},
  number={8},
  pages={1},
  year={2009},
  publisher={Citeseer}
}

@inproceedings{levesque2012winograd,
  title={The winograd schema challenge},
  author={Levesque, Hector and Davis, Ernest and Morgenstern, Leora},
  booktitle={Thirteenth international conference on the principles of knowledge representation and reasoning},
  year={2012}
}

@inproceedings{sanh2019distilbert,
  title={DistilBERT, a distilled version of BERT: smaller, faster, cheaper and lighter.},
  author={Sanh, V},
  booktitle={Proceedings of Thirty-third Conference on Neural Information Processing Systems (NIPS2019)},
  year={2019}
}
%%%%%%%%%%%%%%%%%%%%%%%%%%%%%%%%%%%%%%%%%%%%%%%%%%%%%%%%%%%%

\appendix

%%%%%%%%%%%%%%%%%%%%%%%%%%%%%%%%%%%%%%%%%%%%%%%%%%%%%%%%%%%%

\newpage

\section{Notation Table}
\begin{table*}[!htb]
\caption{Notation table.}
\label{tab:notation}
\centering
\begin{small}
\scalebox{0.9}{
\begin{tabular}{ll|ll}
\toprule
 \textsc{Notation}& \textsc{Description} & \textsc{Notation}& \textsc{Description} \\
  \hline
$\mathbf{w}$ & Model parameter & $\mathbf{g}$ & Gradients \\
$D$, $D_0$, $D_i$ & Dataset & $\eta_i$ & Aggregation weight \\
$\gamma^t$ & Learning rate & $S$ & Set of participation clients \\
$N$ & Participating client size & $M$ & Model dimension \\
$\mathbf{V}^i$, $v_{jk}^i$ & Packed secret shares for gradients & $TS_i$ & Trust score \\
$l$ & \makecell[l]{\# of secrets packed at a polynomial\\ for gradient} & $p$ & \makecell[l]{\# of secrets packed at a polynomial \\for shares of partial dot product} \\
$\mathbf{cs}^i$, $cs_j^i$ &  Shares of partial cosine similarity & $\mathbf{nr}^i$, $nr_j^i$ & Shares of partial gradient norm square \\
$\mathbf{S}^i$, $s_{jk}^i$ & Packed secret shares of $\mathbf{cs}^j$ or $\mathbf{nr}^j$ & $\mathbf{\Tilde{h}}_k^i$ & Secret shares disaggregated along packed index \\
$x_k^i$ & Packed secret share of dot product & $\mathbf{h}_k^i$ & Secret shares aggregated along packed index \\
$e_i$ & Pre-determined secret point & $\alpha_i$ & Pre-selected elements for secret sharing \\
$B_{e_j}$ & \makecell[l]{$n$ by $n$ matrix whose $(i,k)$ \\ entry is $(\alpha_k-e_j)^{i-1}$} & $Chop_d$ & \makecell[l]{$n$ by $n$ matrix whose $(i,k)$ entry is $1$ \\if $1\leq i=k\leq d$ and $0$ otherwise} \\
 \bottomrule
\end{tabular}}
\end{small}
\end{table*}

% \section{Overview of RFLPA}
% Figure \ref{fig:overview} depicts the overview of RFLPA.

\section{Details of Cryptographic Primitives}
\label{app:cryptoprimitive}
\subsection{Packed Shamir Secret Sharing}
The operations of Packed Shamir Secret Sharing performed on a finite field $\mathbb{F}_P$ for some prime number $P$. Denote $\{e_i\}_{i\in [l]}$ as the pre-determined secret point, and $\{\alpha_i\}_{i\in [d]}$ as the pre-selected elements for secret sharing. To share the secrets $\mathbf{g} = (g_1, g_2, ..., g_l)$, the user can generate a degree-$d$ polynomial function:
\begin{equation}
    \phi(x)=q(x) \Pi_{i=1}^l (x-e_i) + \sum_{i=1}^l g_1 L_i(x),
\end{equation}
where $q(x)$ is a random degree-$d-l$ polynomial, and $L_i(x)$ is the Lagrange polynomial $\frac{\Pi_{j\neq i} (x-e_j)}{\Pi_{j\neq i} (e_i-e_j)}$.

The shares sent to player $j$ is generated by:
\begin{equation}
\label{eq:packsss}
    s_{j} = \phi(\alpha_j).
\end{equation}

We use\ $\langle \mathbf{g} \rangle_d$ to denote the degree-$d$ packed secret shares of vector $\mathbf{g}$. The following properties holds for the packed sharing scheme:
\begin{itemize}
    \item $\langle \alpha \mathbf{x}+\beta \mathbf{y} \rangle_d = \alpha \langle \mathbf{x}\rangle_d+\beta \langle \mathbf{y}\rangle_d$
    \item $\langle \mathbf{x}*\mathbf{y} \rangle_{d_1+d_2} = \langle \mathbf{x}\rangle_{d_1}*\langle \mathbf{y}\rangle_{d_2}$
\end{itemize}

\subsection{Key Exchange}
Diffie–Hellman key exchange protocol consists of the following algorithms:
\begin{itemize}
    \item \textit{Generate parameters}: $pp=\mathbf{GenParam}(sp)$ set up the parameters, including prime number and primitive root, according to the security parameter.
    \item \textit{Key generation}: $(s_i^{SK}, s_i^{PK})=\mathbf{KEGen}(pp)$ generates the private-public key pairs for user $i$.
    \item \textit{Key derivation}: $s_{ij} = \mathbf{KEAgree}(s_i^{SK}, s_j^{PK})$ outputs the shared secret key between user $i$ and $j$.
\end{itemize}

\subsection{Symmetric Encryption}
The smmetric encryption scheme consists of the following algorithms:
\begin{itemize}
    \item \textit{Encryption}: $c=\mathbf{Enc}(m,k)$ encrypts message $m$ to cyphertext $c$ using key $k$.
    \item \textit{Decryption}: $m=\mathbf{Dec}(c,k)$ reverses cyphertext $c$ to message $m$ using key $k$.
\end{itemize}
To ensure correctness, we require that $m=\mathbf{Dec}(\mathbf{Enc}(m,k),k)$. For security, the encryption scheme should be indistinguishability under a chosen plaintext attack (IND-DPA) and integrity under ciphertext-only attack (INT-CTXT) \cite{bellare2000authenticated}.

\subsection{Signature Scheme}
The UF-CMA secure signature scheme that consists of a tuple of algorithms $(\mathbf{Gen}, \mathbf{Sign}, \mathbf{Verify})$:
\begin{itemize}
    \item \textit{Key generation}: Based on the security parameter $sp$, $(d^{SK},d^{PK})=\mathbf{SigGen}(sp)$ returns the private-public key pairs.
    \item \textit{Signing algorithm}: $\sigma = \mathbf{Sign}(d^{SK}, m)$ generates a signature $\sigma$ with secret key and message as input.
    \item \textit{Signature verification}: $\mathbf{Verify}(d^{PK}, m, \sigma)$ takes as input the public key, a message and a signature, and returns $1$ if the signature is valid and $0$ otherwise.
\end{itemize}

To proof the security of the signature scheme, we show that no adversary can forge a valid signature on an arbitrary message. Denote a UF-CMA secure signature scheme as $\rm{DS}=(k,\rm{Sign}, \rm{Verify})$, where $k$ is the security parameter. The UF-CMA advantage of an adversary A is defined as $\rm{Adv}_{DS} (A, k)=\mathbb{P}(\rm{Exp}_{\rm{DS}}^{\rm{uf-cma}}(A, k)=1)$, where $\rm{Exp}_{\rm{DS}}^{\rm{uf-cma}} (A,k)$ represents the experiments conducted by adversary A to produce a signature, and $\rm{Exp}_{\rm{DS}}^{\rm{uf-cma}} (A,k)=1$ means that A produced a valid signature. In a UF-CMA secure signature scheme, no probabilistic polynomial time (PPT) adversary is able to produce a valid signature on an arbitrary message with more than negligible probability. In other words, for all PPT adversaries A, there exists a negligible function $\epsilon$ such that $\rm{Adv}_{\rm{DS}} (A,k)\leq \epsilon(k)$.

\section{Comparison between Byzantine-robust aggregation rules}
\label{app:compaggrule}
To provide justification for our algorithm's utilization of FLTrust as the aggregation rule, we summarize the existing Byzantine-robust aggregation rules along four dimensions: (i) computation complexity, (ii) whether the algorithm needs prior knowledge about the number of poisoners, (iii) maximum number of poisoners, (iv) whether the algorithm is compatible with Shamir Secret Sharing (SSS). 
\begin{table*}[!htb]
\caption{Comparison between Byzantine-robust aggregation rules.}
\label{tab:compaggrule}
\centering
\begin{small}
\begin{tabular}{l cccc}
\toprule
 & \textbf{Computation complexity} & \textbf{\makecell{Need prior knowledge\\ about \# of poisoners} }& \textbf{\# of poisoners} & \textbf{\makecell{Compatible \\ with SSS}} \\
  \hline
KRUM & $O(N^2 (M+\log N))$ & Yes & $<50\%$ & Yes \\
 % \hline
Bulyan & $O(N^2 M)$ &	Yes & $<25\%$ &	No \\
 % \hline
Trim-mean & $O(MN \log N)$ & Yes & $<50\%$ & No \\
 % \hline
\rowcolor[gray]{.8} \textbf{FLTrust} & $\mathbf{O(MN))}$ & \textbf{No} & $\mathbf{<100\%}$ & \textbf{Yes} \\
 \bottomrule
\end{tabular}
\end{small}
\end{table*}

Among these dimensions, FLTrust demonstrates clear advantages over other robust aggregation rules:
\begin{itemize}
    \item \textit{Low computation cost:} for a system with $N$ users and $M$ model size, the computation cost of FLTrust is $O(MN)$, lower than existing methods that grow quadratically with $N$.
    \item \textit{No need of prior knowledge about number of poisoners:} the server does not need to know the number of malicious clients in advance to conduct robust aggregation.
    \item \textit{Defend against majority number of poisoners:} benefiting from the trusted root of clean dataset at the server, the aggregation rule we adopted can return robust result even when the number of poisoners is above 50\%.
    \item \textit{Compatible with Shamir Secret Sharing (SSS):} the method we adopted is compatible with the SSS algorithm. While for Bulyan and Trim-mean, there are some non-linear operations not supported by SSS.
\end{itemize}

\section{Algorithm of RFLPA}
\label{app:algorithm}
This section presents the our algorithm to conduct robust federated learning with secure aggregation. 

\begin{algorithm}[htb]
   \caption{RFLPA}
   \label{alg:BRFL}
\begin{algorithmic}
   \STATE \textbf{Input:} Local dataset $D_i$ of clients $i\in [N]$, root dataset $D_0$ at server, number of iterations $T$, security parameter $\kappa$.
   \STATE \textbf{Output:} Global model $\mathbf{w}^T$
   \STATE Clients set up encryption and signature key pairs $(c_i^{PK}, c_i^{SK})$, $(d_i^{PK}, d_i^{SK})$ $\leftarrow$ SetupKeys($N$, $\kappa$) for $i \in [N]$.
   \STATE Server initialize global model $\mathbf{w}^0$
   \FOR{$t\in [1,T]$} 
   \STATE Server conduct local update with root data, compute update norm $\|\mathbf{g}_0\|$, and create packed secret shares $\mathbf{v}_0$.
   \STATE Each clients from $\mathcal{U}_t$ download global model $\mathbf{w}^{t-1}$, corresponding shares of $\mathbf{v}_0$, and $\|\mathbf{g}_0\|$.
   \STATE Server obtain gradients $\mathbf{g}$ $\leftarrow$ RobustSecAgg($\mathcal{U}_t$, $\mathbf{w}^{t-1}$, $\mathbf{v}_0$, $\|\mathbf{g}_0\|$)
   \STATE Server update global model $\mathbf{w}^t \leftarrow\mathbf{w}^{t-1} - \gamma^t \mathbf{g}$
   \ENDFOR
\end{algorithmic}
\end{algorithm}

\begin{algorithm}[htb]
   \caption{SetupKeys}
   \label{alg:setupkey}
\begin{algorithmic}
   \STATE \textbf{Input:} number of clients $N$, security parameter $\kappa$.
   \STATE \textbf{Output:} key pairs $\{(d_i^{PK}, d_i^{SK})\}_{i\in [N]}$;
   \STATE \quad \quad \qquad  secret keys $\{k_{ij}\}_{i,j \in [N]}$.
   \STATE Each user $i\in [N]$ receive their signing key $d_i^{SK}$ from the trusted third party, as well as the verification keys $d_j^{PK}$ of all users $j\in [N]$.
   \STATE Each user $i\in [N]$ generate key pairs $(s_i^{SK},s_i^{PK})=\mathbf{KEGen}(sp)$, and create signature $\sigma_i=\mathbf{Sign}(d_i^{PK}, s_i^{PK})$.
   \STATE Users $i\in [N]$ send $(s_i^{PK}||\sigma_i)$, public key along with signature, to the server.
   \STATE Server distribute $\{(s_i^{PK}||\sigma_i)\}_{i\in [N]}$ to all users.
   \STATE Each user $i$ asserts that $\mathbf{Verify}(d^{PK}, s_j^{PK}, \sigma_j)=1$, and compute $k_{ij} = KEAgree(s_i^{SK}, s_j^{PK})$ for $j\in [N]\backslash i$.
\end{algorithmic}
\end{algorithm}

\begin{algorithm*}[!htb]
   \caption{RobustSecAgg}
   \label{alg:robustsecagg}
\begin{algorithmic}
   \STATE \textbf{Input:} Set of active clients in current iteration $\mathcal{U}_0$, global parameters $\mathbf{w}$ downloaded from server, packed secret shares of server update $\mathbf{v}_0$, norm of server update $\|\mathbf{g}_0\|$.
   \STATE \textbf{Output:} Global aggregated gradient $\mathbf{g}$
   \STATE \textbf{Round 1}:
   \STATE \textit{Client} $i$:
   \STATE \textbullet~ Generate local gradient $\mathbf{g}_i$
   \STATE \textbullet~ Generate packed secrets $\{\mathbf{v}_{ij}\}_{j\in \mathcal{U}_0}$, commitments $\mathcal{C}$ and witness $\{\mathbf{\omega}_{ij}\}_{j\in \mathcal{U}_0}$ for $\mathbf{g}_i$ from \ref{eq:packsss}, \ref{eq:commitment}, and \ref{eq:witness}, encrypt $\mathbf{c}_{ij} = \mathbf{Enc}(\mathbf{v}_{ij}||\mathbf{\omega}_{ij},k_{ij})$, and create signature $\sigma_{ij}=\mathbf{Sign}(d_i^{SK}, \mathbf{c}_{ij}||\mathcal{C})$ for $j\in [N] \backslash i$
   \STATE \textbullet~ Send $(\mathcal{C}||\{\mathbf{c}_{ij}\}_{j\in [N] \backslash i}||\{\sigma_{ij}\}_{j\in [N] \backslash i})$ to the server
   \STATE \textit{Server}:
   \STATE \textbullet~ Collect messages from at least $K$ clients (denote $\mathcal{U}_1$ the set of all respondents).
   \STATE \textbullet~ Send $(\mathcal{C}||\{\mathbf{c}_{ij}\}_{i\in \mathcal{U}_1\backslash j}||\{\sigma_{ij}\}_{i\in \mathcal{U}_1\backslash j})$ to client $j$ for $j\in \mathcal{U}_1$.
   \STATE \textbf{Round 2}:
   \STATE \textit{Client} $i$:
   \STATE \textbullet~ Receive $(\mathcal{C}||\{\mathbf{c}_{ji}\}_{j\in \mathcal{U}_1\backslash i}||\{\sigma_{ji}\}_{j\in \mathcal{U}_1\backslash i})$ from server, and assert that $\mathbf{Verify}(d_j^{PK}, \mathbf{c}_{ji}||\mathcal{C}, \sigma_{ji})=1$.
   \STATE \textbullet~ Recover $(\{\mathbf{v}_{ji}\}_{j\in \mathcal{U}_1\backslash i}, \{\mathbf{\omega}_{ji}\}_{j\in \mathcal{U}_1\backslash i})=\mathbf{Dec}(\mathbf{c}_{ji},k_{ji})$, and verify the secret shares $\{\mathbf{v}_{ji}\}_{j\in \mathcal{U}_1\backslash i}$ by testing \ref{eq:bilineartest}.
   \STATE \textbullet~ Compute local shares of partial norm $\{nr_j^i\}_{j\in \mathcal{U}_1}$ and partial cosine similarity $\{cs_j^i\}_{j\in \mathcal{U}_1}$ from \ref{eq:cosnorm}.
   \STATE \textbullet~ Construct packed secret shares $\{\mathbf{s}_{ik}\}_{k\in \mathcal{U}_1}$,  commitments $\mathcal{C}$, and witness $\{\mathbf{\omega}'_{ik}\}_{k\in \mathcal{U}_1}$ for $(\{nr_j^i\}_{j\in \mathcal{U}_1}, \{cs_j^i\}_{j\in \mathcal{U}_1})$, encrypt $\mathbf{c}'_{ik} = \mathbf{Enc}(\mathbf{s}_{ik}||\mathbf{\omega}'_{ik},k_{ik})$, and create signature $\sigma'_{ik}=\mathbf{Sign}(d_i^{SK}, \mathbf{c}'_{ik}||\mathcal{C})$ for $k\in [N] \backslash i$
   \STATE \textbullet~ Send $(\mathcal{C}||\mathbf{c}'_{ij}||\sigma'_{ij})$ for $j\in [N] \backslash i$ to the server
   \STATE \textit{Server}:
   \STATE \textbullet~ Collect messages from at least $K$ clients (denote $\mathcal{U}_2$ the set of all respondents).
   \STATE \textbullet~ Send $(\mathcal{C}||\{\mathbf{c}'_{ij}\}_{i\in \mathcal{U}_2\backslash j}||\{\sigma'_{ij}\}_{i\in \mathcal{U}_2\backslash j})$ to client $j$ for $j\in \mathcal{U}_2$.
   \STATE \textbf{Round 3}:
   \STATE \textit{Client} $i$:
   \STATE \textbullet~ Receive $(\mathcal{C}||\{\mathbf{c}'_{ji}\}_{j\in \mathcal{U}_2\backslash i}||\{\sigma'_{ji}\}_{j\in \mathcal{U}_2\backslash i})$ from server, and assert that $\mathbf{Verify}(d_j^{PK}, \mathbf{c}'_{ji}||\mathcal{C}, \sigma'_{ji})=1$.
   \STATE \textbullet~ Recover $(\{\mathbf{s}_{ji}\}_{j\in \mathcal{U}_2\backslash i}, \{\mathbf{\omega}'_{ji}\}_{j\in \mathcal{U}_2\backslash i})=\mathbf{Dec}(\mathbf{c}'_{ji},k_{ji})$, and verify the secret shares $\{\mathbf{s}'_{ji}\}_{j\in \mathcal{U}_2\backslash i}$ by testing \ref{eq:bilineartest}.
   \STATE \textbullet~ Obtain the final share of norm $\{\overline{nr}_j^i\}_{j\in |\mathcal{U}_1|/p}$ and cosine similarity $\{\overline{cs}_j^i\}_{j\in |\mathcal{U}_1|/p}$ from \ref{eq:multi2single}, \ref{eq:aggsubgroup}, and Reed-Solomon decoding.
   \STATE \textbullet~ Send $(\{\overline{nr}_j^i\}_{j\in |\mathcal{U}_1|/p}, \{\overline{cs}_j^i\}_{j\in |\mathcal{U}_1|/p}$ to the server.
   \STATE \textit{Server}:
   \STATE \textbullet~ Collect messages from at least $K$ clients (denote $\mathcal{U}_3$ the set of all respondents).
   \STATE \textbullet~ Recover $\{\|\mathbf{g}_j\|^2\}_{j\in \mathcal{U}_1}$ using Reed-Solomon decoding, and assert that $\|\mathbf{g}_j\|^2\leq \|\mathbf{g}_0\|^2$, $\forall j \in \mathcal{U}_1$.
   \STATE \textbullet~ Recover $\{\langle\bar{\mathbf{g}}_i,  \mathbf{g}_0 \rangle\}_{j\in \mathcal{U}_1}$ using Reed-Solomon decoding, and compute the trust score $\{TS_j\}_{j \in \mathcal{U}_1}$ from \ref{eq:trustscore}.
   \STATE \textbullet~ Broadcast the trust score $\{TS_j\}_{j \in \mathcal{U}_1}$ to all users $i \in \mathcal{U}_3$.
   \STATE \textbf{Round 4}: 
   \STATE \textit{Client} $i$:
   \STATE \textbullet~ Compute local aggregation $\langle \mathbf{g} \rangle_i$ from \ref{eq:robustagg}, and send to the server.
   \STATE \textit{Server}:
   \STATE \textbullet~ Collect messages from at least $K$ clients.
   \STATE \textbullet~ Recover $\mathbf{g}$ using Reed-Solomon decoding.
\end{algorithmic}
\end{algorithm*}

Suppose that user $i$ create a packed secret shares $\mathbf{s}$ of $\mathbf{g}$ with polynomial $\phi(x)$. Providing $\kappa$ security, the user sets up generator $\psi$ and secret key $\alpha$, and also outputs the public key $(\psi, \psi^{\alpha}, ..., \psi^{\alpha^d})$ for a degree $d$ polynomial. To make the secret shares verifiable, the user broadcasts a commitment to the function:
\begin{equation}
\label{eq:commitment}
    \mathcal{C} = \psi^{\phi(\alpha)}.
\end{equation}

\section{Verifiable Packed Secret Sharing}
\label{app:vpss}
For each secret $s_l$, user $i$ computes a witness sent to the corresponding client in a private channel:
\begin{equation}
\label{eq:witness}
    w_l = \psi^{(\phi(\alpha)-\phi(l))/(\alpha-l)}.
\end{equation}

After receiving the commitment and witness, user $l$ can verify the secret by checking:
\begin{equation}
\label{eq:bilineartest}
    e(\mathcal{C}, \psi)=e(w_l,\psi^{\alpha}/\psi^l) e(\psi,\psi)^{\phi(l)},
\end{equation}
where $e(\cdot)$ denotes a symmetric bilinear pairing. 

The correctness and secrecy of the protocol are guarantee by the discrete logarithm (DL) \cite{alfred1997handbook}, $t$-polynomial Diffie-Hellman ($t$-polyDH) \cite{kate2010constant}, and $t$-Strong Diffie-Hellman ($t$-SDH) 
\cite{boneh2004short} assumptions. 

\section{Explanation of Secret Re-sharing}
\label{app:secretreshare}
For $m\in [\lceil N/p \rceil]$, the shares of secret $cs_{(m-1)p+k}^i$ for some $k\in [p]$ can be represented as:
\begin{equation}
\begin{gathered}
\setstacktabbedgap{3pt}
    \parenMatrixstack{
    s_{1m}^i & \dots & s_{Nm}^i
    } = \parenMatrixstack{
    cs_{(m-1)p+k}^i & \theta_1 & \dots & \theta_d & 0 & \dots & 0
    } \times B_{e_k},
\end{gathered}
\end{equation}
where $\{\theta_j\}_{j\in [d]}$ are random integers.

Hence, the user side computation of \ref{eq:multi2single} is the same as:
\begin{equation}
\begin{gathered}
\setstacktabbedgap{3pt}
    \parenMatrixstack{
    s_{1m}^1 & \dots & s_{1m}^N\\
    \vdots & \ddots & \vdots\\
    s_{Nm}^1 & \dots & s_{Nm}^N
    } B_{e_j}^{-1} Chop_d B_{e_j'} = B_{e_k}^T\\
    \times \parenMatrixstack{
    cs_{(m-1)p+k}^1 &\dots & cs_{(m-1)p+k}^N\\
    \vdots & \ddots & \vdots
    } B_{e_j}^{-1} Chop_d.
\end{gathered}
\end{equation}

The aggregation of new secret and reconstruction of $\{x_m^j\}$ is equivalent to taking the first column of:
\begin{equation}
\begin{gathered}
\setstacktabbedgap{3pt}
    B_{e_k}^T\parenMatrixstack{
    cs_{(m-1)p+k}^1 &\dots & cs_{(m-1)p+k}^N\\
    \vdots & \ddots & \vdots
    } \\
    \times \left(B_{e_1}^{-1}+\dots+ B_{e_l}^{-1}\right) Chop_d.
\end{gathered}
\end{equation}
Since $\mathbf{cs}^j$ is a packed secret share of the partial cosine similarity, it follows that:
\begin{equation}
\begin{gathered}
\setstacktabbedgap{3pt}
    \parenMatrixstack{
    cs_{h}^1 &\dots & cs_{h}^N} B_{e_j}^{-1} Chop_d = \parenMatrixstack{\sum_{(j-1)l< i \leq jl} \bar{g}_{hi}g_{0i} & \dots},
\end{gathered}
\end{equation}
meaning that the first elements gives the partial cosine similarity.

Therefore, the final shares sent to server $\{x_m^j\}$ can be formulated as:
\begin{equation}
\setstacktabbedgap{3pt}
     \parenMatrixstack{
    x_m^1 &\dots & x_m^N} = \parenMatrixstack{\sum_{i} \bar{g}_{m(p-1)+h,i} g_{0i} & \theta_1 & \dots & \theta_d & 0 & \dots} B_{e_h},
\end{equation}
for $h\in (m(p-1), mp]$. Therefore, the server could retrieve the dot product by Reed-Solomon decoding, which is equivalent to multiplying $\{B_{e_h}^{-1}\}_{h\in (m(p-1), mp]}$ and obtaining the first element.

\section{Details of Complexity Analysis}
\label{det:complexity analysis}

\textbf{User computation}: User's computation cost can be broken as: (1) generating packed secret shares of update ($O(M+N)\log^2 N)$ complexity \cite{kung1973fast}); (2) computing shares of partial gradient norm square and cosine similarity ($O(M+N))$ complexity); (3) creating packed secret shares of partial gradient norm square and cosine similarity ($O(N \log^2 N)$ complexity); (4) deriving final secret shares of gradient norm square and cosine similarity ($O(N^2\log^2 N)$ complexity). Therefore, each user's computation cost is $O((M+N^2)\log^2 N)$.

\textbf{User communication}: User's communication cost can be broken as: (1) downloading parameters from server ($O(M)$ messages); (2) sending and receiving secret shares of gradient ($O((M, N))$ messages); (3) sending and receiving secret shares of partial gradient norm square and cosine similarity ($O(N)$ messages); (4) sending final shares of gradient norm square and cosine similarity ($O(1)$ messages); (5) receiving trust scores from the server ($O(N)$ messages); (6) sending shares of aggregated update to the server ($O(M/N+1)$ messages). Hence, each user's communication cost is $O(M+N)$.

\textbf{Server computation}: The server's computation cost can be broken as: (1) recovering gradient norm square and cosine similarity by Reed-Solomon decoding ($O(N \log^2 N \log \log N)$ complexity \cite{gao2003new}); (2) computing the trust score of each user ($O(N)$ complexity); (3) decoding the aggregated global gradient ($O(M+N)\log^2 N \log \log N)$ complexity). Therefore, the server's computation cost is $O((M+N) \log^2 N \log \log N)$.

\textbf{Server communication}: The server's communication cost can be broken as: (1) distributing parameters to clients ($O(MN)$ messages); (2) sending and receiving secret shares of user update ($O((M+N)N\})$ messages); (3) sending and receiving secret shares of partial gradient norm square and cosine similarity ($O(N^2)$ messages); (4) receiving final shares of gradient norm square and cosine similarity ($O(N)$ messages); (5) broadcasting trust scores to clients ($O(N^2)$ messages); (6) receiving shares of aggregated update from clients ($O(M+N)$ messages). Overall, the server's communication cost is $O((M+N)N)$.

\section{Proof of Theorem \ref{thm:security}}
\label{app:security}
\begin{proof}
    We utilize the standard hybrid argument to prove the theorem. we define a PPT simulator {\rm SIM} through a series of (polynomially many) subsequent to ${\rm REAL}_{\mathcal{C}}^{\mathcal{U},t,\kappa}$, so that the view of $\mathcal{C}$ in {\rm SIM} is computationally indistinguishable from that in ${\rm REAL}_{\mathcal{C}}^{\mathcal{U},t,\kappa}$.

    ${\rm Hyb_1}$: In the hybrid, each honest user from $\mathcal{U}_1\backslash \mathcal{C}$ encrypts shares of a uniformly random vector, instead of the raw gradients. The properties of Shamir's secret sharing ensure that the distribution of any $|\mathcal{C}\backslash \{S\}|<t$ shares of raw gradients is identical to that of any equivalent length vector, and IND-CPA security guarantees that the view of server is indistinguishable in both cases. Hence, this hybrid is identical from the previous one.

    ${\rm Hyb_2}$: In the hybrid, the simulator aborts if $\mathcal{C}$ provides any of the honest user $i$ with a signature on $j$'s message, $\mathbf{c}_{ji}$, but the user couldn't produce the same signature given the public key (in round 2). The security of the signature scheme guarantees that this hybrid is indistinguishable from the previous one.

    ${\rm Hyb_3}$: In this hybrid, {\rm SIM} aborts if any of the honest user $i$ fails to verify the secret shares $\mathbf{s}_{ji}$ from user $j$ by checking \ref{eq:bilineartest}. The the DL, $t$-polyDH, and $t$-SDH assumptions guarantee that this hybrid is identical from the previous one.

    ${\rm Hyb_4}$: In the hybrid, each honest user from $\mathcal{U}_2\backslash \mathcal{C}$ encrypts shares of a uniformly random vector rather than partial norm and cosine similarity. The properties of Shamir's secret and IND-CPA security ensure that this hybrid is indistinguishable from the previous one.

    ${\rm Hyb_5}$: In the hybrid, the simulator aborts if $\mathcal{C}$ provides any of the honest user $i$ with a signature on $j$'s message, $\mathbf{c}'_{ji}$, but the user couldn't produce the same signature given the $j$'s key (in round 3). Because of the security of the signature scheme, this hybrid is indistinguishable from the previous one.

    ${\rm Hyb_6}$: This hybrid is defined as ${\rm Hyb_3}$, with the only difference that {\rm SIM} verify the secret shares $\mathbf{s}'_{ji}$ in round 3. This hybrid is indistinguishable from the previous one under DL, $t$-polyDH, and $t$-SDH assumptions.

    The above changes do not modify the views seen by the colluding parties, and the hybrid doesn't make use of the honest users' input. Therefore, the output of {\rm SIM} is computationally indistinguishable from the output of ${\rm REAL}_{\mathcal{C}}^{\mathcal{U},t,\kappa}$, and this concludes the proof.
    
\end{proof}

\section{Proof of Theorem \ref{convergence}}
Denote $\mathbf{\bar{g}}^t = \sum_i \eta_i \mathbf{\bar{g}}_i$ be the aggregated gradients at iteration $t$.
\begin{lemma}
For arbitrary number of adversarial clients, the distance between $\mathbf{\bar{g}}^t$ and $\nabla F(\mathbf{w}^t)$ is bounded by:
\begin{equation}
    \|\mathbf{\bar{g}}^t-\nabla F(\mathbf{w}^t)\| \leq 3\|\mathbf{g}_0^t - \nabla F(\mathbf{w}^t) \| + 2\|\nabla F(\mathbf{w}^t)\| + \frac{\sqrt{d}}{q}.
\end{equation}
\end{lemma}
\begin{proof}
It follows that:
\begin{equation}
\begin{gathered}
    \|\mathbf{\bar{g}}^t-\nabla L^t(\mathbf{w})\| = \|\sum_i \eta_i \mathbf{\bar{g}}_i-\nabla F^t(\mathbf{w})\|\\
    =\|\sum_i \eta_i \mathbf{\bar{g}}_i-\mathbf{\bar{g}}_0+\mathbf{\bar{g}}_0-\mathbf{g}_0+\mathbf{g}_0-\nabla F^t(\mathbf{w})\|\\
    \leq \|\sum_i \eta_i \mathbf{\bar{g}}_i-\mathbf{\bar{g}}_0\| + \|\mathbf{\bar{g}}_0-\mathbf{g}_0\| + \|\mathbf{g}_0-\nabla F^t(\mathbf{w})\|\\
    \leq \sum_i \eta_i\|\mathbf{\bar{g}}_i\| + \|\mathbf{\bar{g}}_0\|+\|\mathbf{\bar{g}}_0-\mathbf{g}_0\|+\|\mathbf{g}_0-\nabla F^t(\mathbf{w})\|\\
    \overset{(a)}{\leq} 2 \|\mathbf{g}_0\| + \frac{\sqrt{d}}{q}+\|\mathbf{g}_0-\nabla F^t(\mathbf{w})\|\\
    \leq 3\|\mathbf{g}_0-\nabla F^t(\mathbf{w})\| + 2\|\nabla F^t(\mathbf{w})\|+\frac{\sqrt{d}}{q},
\end{gathered}
\end{equation}
where $(a)$ is because $\sum_i \eta_i = 1$, $\|\mathbf{\bar{g}}_i\|\leq \|\mathbf{g}_0\|$, and$\|\mathbf{\bar{g}}_0\|\leq \|\mathbf{g}_0\|$.
\end{proof}
\begin{lemma}
Under Assumption \ref{converge_ass1}, we have the following bound at iteration $t$:
\begin{equation}
    \|\mathbf{w}^{t}-\mathbf{w}^*-\gamma \nabla F(\mathbf{w}^t)\| \leq \sqrt{1-\mu^2/(4L_g^2)} \|\mathbf{w}^{t}-\mathbf{w}^*\|.
\end{equation}
\end{lemma}
\begin{proof}
Refer to lemma 2 in \cite{cao2021fltrust} for the proof.
\end{proof}

\begin{lemma}
Suppose Assumption \ref{converge_ass1}, \ref{converge_ass2}, \ref{converge_ass3} holds. For any $\delta \in (0,1)$, if $\Delta_1\leq \nu_1^2/\alpha_1$, $\Delta_2\leq \nu_2^2/\alpha_2$, we have:
\begin{equation}
    P\left\{\|\mathbf{g}_0 - \nabla F(\mathbf{w})\|\leq 8\Delta_2 \|\mathbf{w}-\mathbf{w}^*+4\Delta_1\|\right\}\geq 1-\delta,
\end{equation}
for any $\mathbf{w}\in \Theta \subset \left\{\mathbf{w}:\|\mathbf{w}-\mathbf{w}^*\|\leq r\sqrt{d}\right\}$ given some positive number $r$.
\end{lemma}
\begin{proof}
Refer to lemma 4 in \cite{cao2021fltrust} for the proof.
\end{proof}

\textbf{Proof of Theorem \ref{convergence}}: Given the lemmas above, we can proceed to prove Theorem \ref{convergence}. We have:
\begin{equation}
\begin{gathered}
    \|\mathbf{w}-\mathbf{w}^*\| \leq \|\mathbf{w}^{t-1}-\gamma \nabla F(w^{t-1})-\mathbf{w}^*\|+\gamma\|\mathbf{\bar{g}}^t-\nabla F(\mathbf{w}^t)\|\\
    \leq \|\mathbf{w}^{t-1}-\gamma \nabla F(w^{t-1})-\mathbf{w}^*\| + 3\gamma \|\mathbf{g}_0^t - \nabla F(\mathbf{w}^t) \|\\
    + 2\gamma \|\nabla F(\mathbf{w}^t)\| + \frac{\gamma\sqrt{d}}{q}\\ \leq \left(\sqrt{1-\mu^2/(4L^2)}+24\gamma \Delta_2+2\gamma L\right) \|\mathbf{w}^{t-1}-\mathbf{w}^*\|\\
    +12\gamma \Delta_1 + \frac{\gamma\sqrt{d}}{q}.
\end{gathered}
\end{equation}
Therefore, with probability at least $1-\delta$, it follows that:
\begin{equation}
    \|\mathbf{w}^t-\mathbf{w}^*\| \leq (1-\rho)^t \|\mathbf{w}^0-\mathbf{w}^*\| + +12\gamma \Delta_1 + \frac{\gamma\sqrt{d}}{q}.
\end{equation}

\section{Assumptions for convergence analysis \ref{sec:convergence}}
\label{app:converge_ass}

\begin{assumption}
\label{converge_ass1}
The expected risk function $F(\mathbf{w})$ is $\mu$-strongly convex and $L$-smooth for any $\mathbf{w}$, $\mathbf{\bar{w}}$:
\begin{equation}
\begin{gathered}
    F(\mathbf{\bar{w}}) \geq F(\mathbf{w}) + \langle \nabla F(\mathbf{w}), \mathbf{\bar{w}}-\mathbf{w} \rangle +\frac{\mu}{2} \|\mathbf{\bar{w}}-\mathbf{w}\|^2\\
    \|\nabla F(\mathbf{w}) - \nabla F(\mathbf{\bar{w}})\| \leq L \|\mathbf{\bar{w}}-\mathbf{w}\|.
\end{gathered}
\end{equation}
Moreover, the empirical loss function $L(D, \mathbf{w})$ is $L_1$-smooth probabilistically. For any $\delta \in (0,1)$, there exists an $L_1$ such that:
\begin{equation}
    P\left\{ \sup_{\mathbf{w}\neq \mathbf{\bar{w}}} 
    \frac{\|\nabla L(D, \mathbf{w})-\nabla L(D, \mathbf{\bar{w}})\|}{\|\mathbf{w}-\mathbf{\bar{w}}\|}\leq L_1\right\} \geq 1-\frac{\delta}{3}.
\end{equation}
\end{assumption}
\begin{assumption}
\label{converge_ass2}
The root dataset $D_0$ and clients' local dataset $D_i(i=1,2,...,n)$ are sampled independently from distribution $\chi$.
\end{assumption}
\begin{assumption}
\label{converge_ass3}
The gradients of the empirical loss function $\nabla L(D, \mathbf{w}^*)$ at the optimal model $\mathbf{w}^*$ is bounded. Furthermore, $h(D, \mathbf{w})=\nabla L(D, \mathbf{w})-\nabla L(D, \mathbf{w}^*)$ is also bounded. Specifically, $\langle \nabla L(D, \mathbf{w}^*), \mathbf{v} \rangle$ and $\langle h(D, \mathbf{w}) - \mathbb{E}[h(D, \mathbf{w})],  \mathbf{v}\rangle/\|\mathbf{w}-\mathbf{w}^*\|$ are sub-exponential for any unit vector $\mathbf{v}$. Formally, for $\forall |\lambda|\leq 1/\alpha_1$, $\forall |\lambda|\leq 1/\alpha_2$, $\mathbf{B}=\{\mathbf{v}:\|v\|=1\}$, it holds that:
\begin{equation}
\begin{gathered}
    \sup_{\mathbf{v}\in \mathbf{B}} \mathbb{E}[\exp(\lambda \langle \nabla L(D, \mathbf{w}^*), \mathbf{v} \rangle)] \leq e^{\nu_1^2 \lambda^2/2}\\
    \sup_{\mathbf{v}\in \mathbf{B}, \mathbf{w}} \mathbb{E}\left[\exp\left(\frac{\langle h(D, \mathbf{w}) - \mathbb{E}[h(D, \mathbf{w})],  \mathbf{v}\rangle}{\|\mathbf{w}-\mathbf{w}^*\|} \right)\right] \leq e^{\nu_2^2 \lambda^2/2}.
\end{gathered}
\end{equation}
\end{assumption}

\section{Experiments}
The experiments are conducted on a 16-core Ubuntu Linux 20.04 server with 64GB RAM and A6000 driver, where the programming language is Python.
\subsection{Datasets}
MNIST is a collection of handwritten digits, including 60,000 training and 10,000 testing images of $28\times28$ pixels. F-MNIST consists of 70,000 fashion images of size $28\times 28$ and is split into 60,000 training and 10,000 testing samples. CIFAR-10 is natural dataset that includes 60,000 $32\times32$ colour images in 10 classes, splitting into 50,000 training and 10,000 testing images.
\subsection{FL configuration}
\label{app:flconfig}
both datasets are split among 10,000 users and select 100 users in each iteration. The server stores $200$ clean samples as benchmark. We allow up to $20\%$ clients to drop out in each round, and a maximum of $30\%$ participating clients to collaborate with each other to reveal the secret. Therefore, we construct a secret sharing of degree $40$, considering the doubling of degree during dot product computation, and pack each $10$ elements into a secret.

\subsection{Hyper-Parameters}
The parameters are updated using Adaptive Moment Estimation (Adam) method with a learning rate of $0.01$. Each accuracy reported in the tables is an average of $5$ experiments, and each round of experiments runs for $200$ iterations. Both LDP and CDP adopt privacy parameter $\epsilon=3$ and $\delta=0.0001$.

\subsection{Comparison among Aggregation Frameworks}
\label{app:compareframe}
In Table \ref{tab:compareframe} we summarize the comparison among aggregation frameworks along four dimensions:
\begin{itemize}
    \item \textit{Robustness against malicious users:} most algorithms provide certain level of robustness against malicious users. Local DP is not that effective in defending malicious users according to our experiment results. Though Robust Federated Aggregation (RFA) \cite{pillutla2022robust} provides a robust aggregation protocol based on geometric median, the malicious users could freely manipulate the uploaded gradients for poisoning attacks.
    \item \textit{Privacy Protection against server:} whether the framework protect user’s plaintext gradient against server. Only PEFL, PBFL, ShieldFL, SecureFL \cite{hao2021efficient}, RoFL \cite{lycklama2023rofl}, ELSA \cite{rathee2023elsa}, BREA, and RFLPA achieves the goals of robustness and privacy simultaneously.
    \item \textit{Collusion threshold during model training:} the server could obtain users’ plaintext gradients if it colludes with more than the given level of parties. PEFL, PBFL, ShieldFL, SecureFL, and ELSA all rely on two non-colluding parties during model training to protect users’ message. The collaboration between the two non-colluding parties could compromise user’s privacy.
    \item \textit{MPC techniques:} the main multiparty computation techniques leveraged by the framework. PEFL, PBFL, ShieldFL, SecureFL, and ELSA are based on multi-party computation (MPC) or homomorphic encryption (HE), RoFL is based on zero-knowledge proof (ZKP), and BREA and RFLPA are based on secret sharing.
\end{itemize}
Furthermore, although RoFL and ELSA could defend against malicious users, they are designed specifically for a naive robust aggragation method, norm bounding. It's completely impractical to generalize these frameworks to more advance defense method such as Krum.

\begin{table*}[!htb]
\caption{Corse-grained comparison among Aggregation Frameworks. “/” denotes non-applicable. ELSA improves on RoFL regarding the the efficiency.}
\label{tab:compareframe}
\centering
% \begin{scriptsize}
\begin{small}
\begin{tabular}{l cccc}
\toprule
& \makecell{Robustness against \\malicious users} & \makecell{Privacy Protection\\ against server} & \makecell{Collusion threshold\\ during model training} & MPC techniques \\
  \hline
FedAvg & Yes & No & / & / \\

Bulyan & Yes & No & / & / \\

Trim-mean & Yes & No & / & / \\

KRUM & Yes & No & / & / \\

Central DP & Yes & No & / & / \\
\hline
Local DP & Not effective & Yes & / & / \\
RFA & No & Yes & / & / \\
\hline
PEFL & Yes & Yes & 1 & 	HE (Paillier) \\

PBFL & Yes & Yes & 1 & 	HE (CKKS) \\

ShieldFL & Yes & Yes & 1 & 	HE (Paillier) \\

SecureFL & Yes & Yes & 1 & MPC \& HE (BFV) \\
\hline
RoFL & Yes & Yes & $O(N)$ & ZKP\\
ELSA & Yes & Yes & 1 & 	MPC\\
\hline

BREA & Yes & Yes & $O(N)$ & Secret sharing \\

\rowcolor[gray]{.8} RFLPA & Yes & Yes & $O(N)$ & Secret sharing \\
 \bottomrule
\end{tabular}
% }
\end{small}
\end{table*}

\subsection{Accuracies over Iterations}
\label{app:acciter}
Figure \ref{fig:accuracy} demonstrates the impact of different iterations on test accuracies for RFLPA, BREA and FedAvg using the MNIST dataset. The results reveal that the RFLPA algorithm displays comparable convergence regardless of the existence of attackers, while FedAvg exhibits significantly inferior convergence when $30\%$ attackers are present.

\begin{figure}[ht]
\begin{center}
  \centering    \centerline{\includegraphics[width=1\columnwidth]{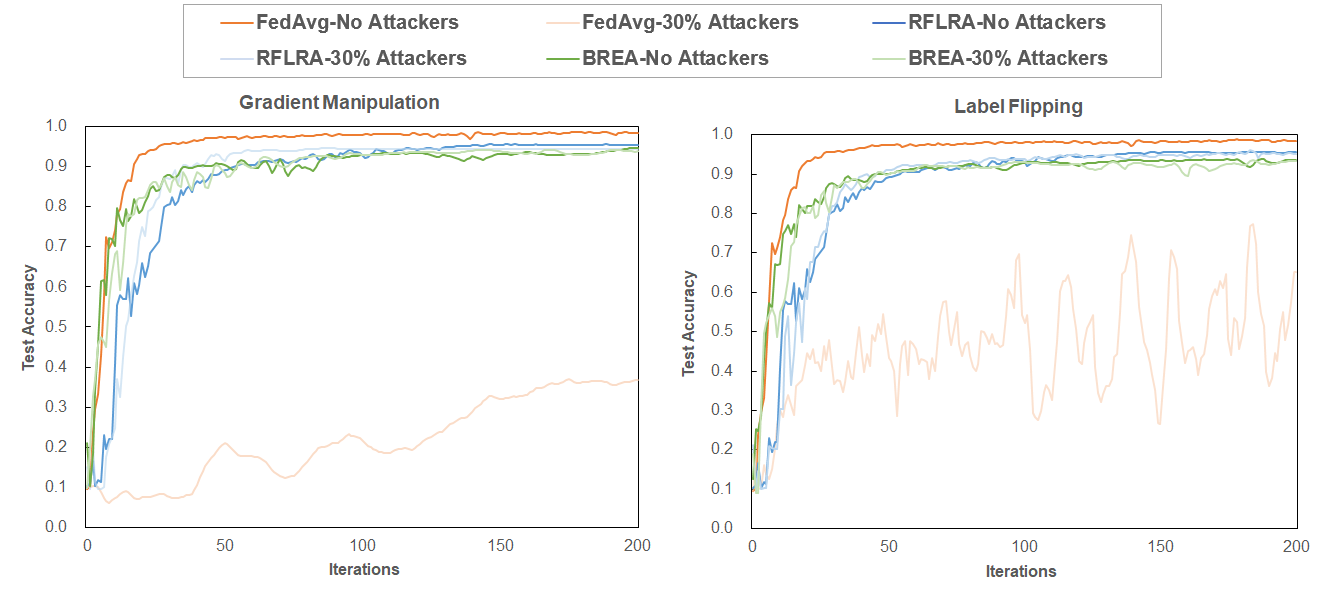}}
\caption{Test accuracy of RFLPA and FedAvg for different proportions of malicious users on MNIST dataset.}
\label{fig:accuracy}
\end{center}
\end{figure}

\subsection{Performance on Additional Attacks}
\label{app:moreattack}
\subsubsection{Poisoning Attacks}
We evaluate our protocol against several stealthier attacks: (1) KRUM attack \cite{fang2020local}, (2) BadNets \cite{gu2019badnets}, and (3) Scaling attack \cite{bagdasaryan2020backdoor}. KRUM attack is untarget attack, and BadNets as well as Scaling attack are backdoor attacks that specifically degrade the performance on triggered samples. We follow the same approach as in \cite{gu2019badnets} and \cite{bagdasaryan2020backdoor} to embed triggers in the targeted images.

Table \ref{tab:backdoor} compares the performance of RFLPA and FedAvg against the above attacks. For KRUM attack, RFLPA improves the accuracy on the general dataset over FedAvg by more than 1.6x. For the two backdoor attacks, RFLPA show trivial performance loss on the general and triggered dataset, as opposed to the significant degradation in accuracy for FedAvg.
\begin{table*}[!htb]
\caption{Accuracies on CIFAR-10 under varying proportions of attackers. For backdoor attacks, the values are presented as \textit{overall accuracy (backdoor accuracy)}.}
\label{tab:backdoor}
\centering
% \begin{scriptsize}
\scalebox{0.95}{
\begin{tabular}{l|ccc|ccc}
\toprule
  & \multicolumn{3}{c|}{FedAvg} & \multicolumn{3}{c}{RFLPA} \\
 \% of attackers & 10\% & 20\% & 30\% & 10\% & 20\% & 30\% \\
  \hline
KRUM attack & 0.27 & 0.12 & 0.11 & 0.71 & 0.70 & 0.70 \\
BadNets & 0.68 (0.54) & 0.67 (0.54) & 0.55 (0.28) & 0.71 (0.68) & 0.70 (0.68) & 0.69 (0.66) \\
Scaling attack & 0.70 (0.22) & 0.68 (0.21) & 0.54 (0.19) & 0.70 (0.69) & 0.70 (0.69) & 0.69 (0.69) \\
 \bottomrule
\end{tabular}
}
\end{table*}

\subsubsection{Inference Attacks}
We assess our RFLPA against passive inference attack using the Deep Leakage from Gradients (DLG) \cite{zhu2019deep}. It is important to note that experiments were not conducted for active inference attacks, where the server might alter users’ messages, such as secret shares, to access private data. This omission is due to the protection provided by the signature scheme, which safeguards message integrity and prevents the server from forging any user's messages.

DLG attempts to reconstruct the original image from the aggregated gradients. We conducted an attack on the CIFAR-10 dataset, using the specifications in Appendix \ref{app:flconfig}. The average peak signal-to-noise ratio (PSNR) of generated image with respect to original image is 11.27, much lower than the value of 36.5 when no secure aggregation is involved. Figure \ref{fig:inferatk} shows that the inferred images are far from the raw images under DLG attack.

\begin{figure}[!htp]
\begin{center}
\includegraphics[width=100mm]{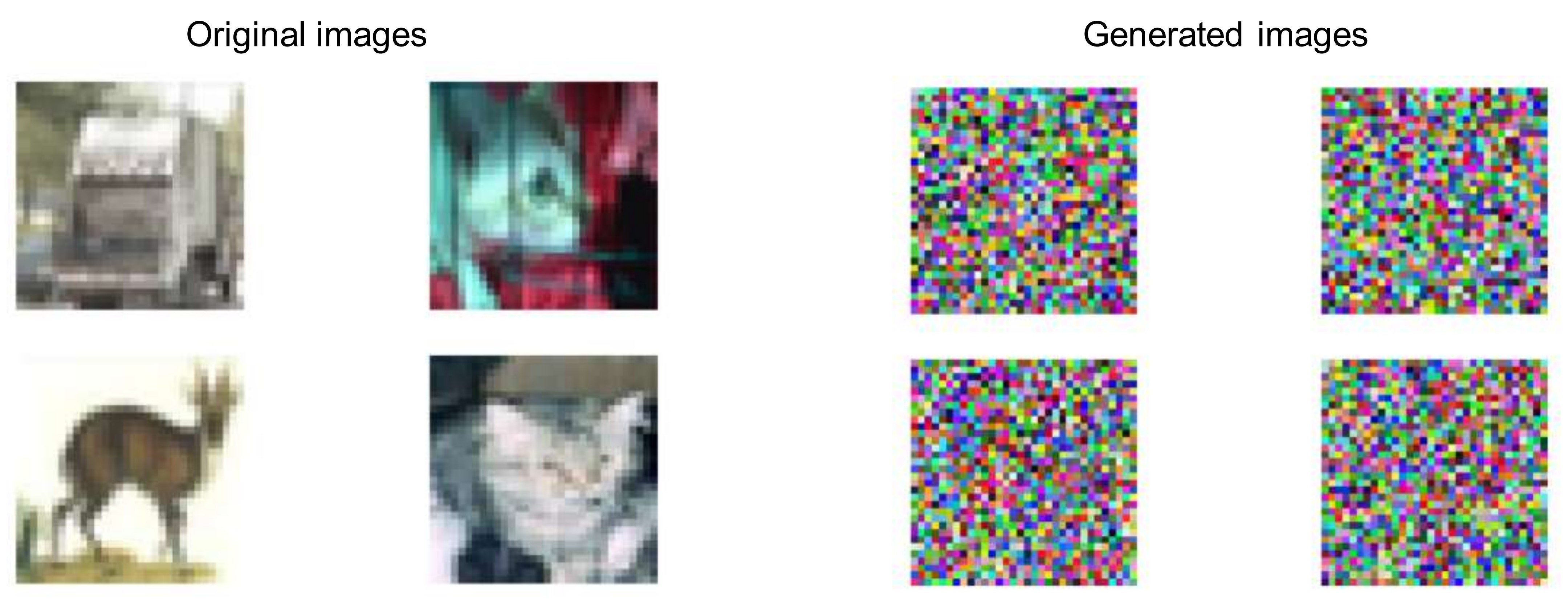}
\caption{Original and inferred image under RFLPA.}
\label{fig:inferatk}
\end{center}
\end{figure}

\subsection{Performance on Diverse Dataset}
\label{app:nlpcifar}
\subsubsection{Performance on Natural Language Processing (NLP) Dataset}
\label{app:nlp}
We evaluate the accuracy of our framework on two NLP datasets, Recognizing Textual Entailment (RTE) \cite{bentivogli2009fifth} and Winograd NLI (WNLI) \cite{levesque2012winograd}, by finetuning a distillBERT model \cite{sanh2019distilbert}. We present the performance for gradient manipulation attack in Table \ref{tab:accnlp}. The result demonstrates that for the two NLP datasets, RFLPA has robust accuracies in the presence of up to 30\% attackers.

\begin{table*}[!htb]
\caption{Accuracies on NLP dataset under different proportions of attackers.}
\label{tab:accnlp}
\centering
\begin{small}
\scalebox{0.95}{
\begin{tabular}{l|cccc|cccc}
\toprule
 & \multicolumn{4}{c|}{RTE} & \multicolumn{4}{c}{WNLI} \\
  \hline
Proportion of Attackers & No & $10\%$ & $20\%$ & $30\%$ & No & $10\%$ & $20\%$ & $30\%$ \\
\hline
FedAvg & $\mathbf{0.599}$   & $0.509$ & $0.487$ & $0.462$ & $\mathbf{0.619}$ & $0.563$ & $0.437$ & $0.437$ \\

BREA & $0.584$   & $\mathbf{0.592}$ & $0.570$ & $0.567$ & $0.592$ & $\mathbf{0.592}$ & $0.577$ & $0.563$ \\

\rowcolor[gray]{.8} RFLPA & $0.596$   & $0.582$ & $\mathbf{0.582}$ & $\mathbf{0.577}$ & $0.619$ & $0.592$ & $\mathbf{0.592}$ & $\mathbf{0.563}$\\
 \bottomrule
\end{tabular}
}
\end{small}
\end{table*}

\subsubsection{Performance on CIFAR-100 Dataset}
To test a more complex CV dataset, we evaluate our frameworks on CIFAR-100 \cite{krizhevsky2009learning} dataset using a ResNet-9 classifier. It can be observed in Figure \ref{tab:cifar100} that RFLPA significantly enhances the accuracy over FedAvg from 10\% attackers, by an average of 3.94x. Furthermore, RFLPA experiences little performance degradation in the presence of up to 30\% attackers.

\begin{table*}[!htb]
\caption{Accuracy on CIFAR-100 dataset under gradient manipulation attack.}
\label{tab:cifar100}
\centering
% \begin{scriptsize}
\scalebox{0.86}{
\begin{tabular}{l|cccc|cccc}
\toprule
  & \multicolumn{4}{c|}{FedAvg} & \multicolumn{4}{c}{RFLPA} \\
 \% of attackers & No & 10\% & 20\% & 30\% & No & 10\% & 20\% & 30\% \\
  \hline
Accuracy & $0.55$ \tiny$\mathbf{\pm 0.2}$\normalsize & $0.11$ \tiny$\mathbf{\pm 0.5}$\normalsize & $0.10$ \tiny$\mathbf{\pm 0.0}$\normalsize & $0.10$ \tiny$\mathbf{\pm 0.0}$\normalsize & $0.55$ \tiny$\mathbf{\pm 0.2}$\normalsize & $0.54$\tiny$\mathbf{\pm 0.1}$\normalsize & $0.50$ \tiny$\mathbf{\pm 0.1}$\normalsize & $0.49$ \tiny$\mathbf{\pm 0.2}$\normalsize \\
 \bottomrule
\end{tabular}
}
\end{table*}

\subsection{Overhead Analysis}
\label{app:overhead}
\subsubsection{Computation Time between RFLPA and HE-based methods}
\label{app:overheadhe}
To verify the practicability of RFLPA, we benchmark our framework with three HE-based methods, PEFL \cite{liu2021privacy}, PBFL \cite{miao2022privacy}, and ShieldFL \cite{ma2022shieldfl}. Table \ref{tab:compcosthe} presents the per-iteration computation time using a MNIST classifier (1.6M parameters) for the three algorithms and RFLPA. It can be observed that it takes 1.5 to 6.5 day to run the three HE-based algorithms for only a single iteration, which renders them impractical for real-life deployment.

\begin{table*}[!htb]
\caption{Computation cost (in minutes) with varying client size.}
\label{tab:compcosthe}
\centering
% \begin{scriptsize}
\begin{tabular}{l|cccc|cccc}
\toprule
 & \multicolumn{4}{c|}{Per-user Cost} & \multicolumn{4}{c}{Server Cost} \\
  \hline
Client size & $100$ & $200$ & $300$ & $400$ & $100$ & $200$ & $300$ & $400$ \\
 \hline
\rowcolor[gray]{.8} RFLPA & $3.41$ & $11.44$ & $24.51$ & $42.60$  & $6.68$ & $8.46$ & $15.00$ & $26.47$ \\
 % \hline
PEFL & $111.51$ & $109.27$ & $109.44$ & $110.13$  & $2156.20$ & $6056.98$ & $6785.71$ & $9365.46$ \\
 % \hline
PBFL & $12.65$ & $12.58$ & $12.73$ & $12.63$  & $1806.05$ & $3598.54$ & $5386.97$ & $7193.64$ \\
 % \hline
 ShieldFL & $111.73$ & $109.43$ & $109.25$ & $109.84$  & $2192.48$ & $6093.05$ & $6809.60$ & $9384.11$ \\
 \bottomrule
\end{tabular}
% \end{scriptsize}
\end{table*}
\subsubsection{Ablation Study}
\label{app:ablation}
Considering that RFLPA and BREA leverage different robust aggregation rule, we conducted ablation study to demonstrate that the reduction in overhead is attributed to the scheme design of RFLPA rather than the inherent advantages of the underlying aggregation rule. In particular, we replace the aggregation module in RFLPA with KRUM, and presents the per-iteration communication and computation cost, respectively, in Table \ref{tab:comcostablation} and \ref{tab:comtimeablation}. It can be observed that even with substituting the aggregation module with KRUM in our framework, there’s still notable reduction in the communication cost benefiting from the design of our secret sharing algorithm.

% \definecolor{lightb}{RGB}{217,224,250}
\begin{table*}[!htb]
\caption{Communication cost (in MB) per client with varying client size with MNIST classifier (1.6M parameters). RFLPA (KRUM) replaces the aggregation rule with KRUM in RFLPA.}
\label{tab:comcostablation}
\centering
% \begin{scriptsize}
\begin{tabular}{l cccc}
\toprule
Client size & $300$ & $400$ & $500$ & $600$  \\
 \hline
\rowcolor[gray]{.9} RFLPA & $82.51$ & $82.52$ & $82.53$ & $82.54$ \\

BREA & $1909.92$ & $2544.45$ & $3178.98$ & $3813.51$ \\

\rowcolor[gray]{.7} RFLPA (KRUM) & $79.58$ & $82.25$ & $85.68$ & $89.87$\\

 \bottomrule
\end{tabular}
% \end{scriptsize}
\end{table*}

\begin{table*}[!htb]
\caption{Computation cost (in minutes) with varying client size with MNIST classifier (1.6M parameters). RFLPA (KRUM) replaces the aggregation rule with KRUM in RFLPA.}
\label{tab:comtimeablation}
\centering
% \begin{scriptsize}
\begin{tabular}{l|cccc|cccc}
\toprule
 & \multicolumn{4}{c|}{Per-user Cost} & \multicolumn{4}{c}{Server Cost} \\
  \hline
Client size & $100$ & $200$ & $300$ & $400$ & $100$ & $200$ & $300$ & $400$ \\
 \hline
\rowcolor[gray]{.9} RFLPA & $3.41$ & $11.44$ & $24.51$ & $42.60$  & $6.68$ & $8.46$ & $15.00$ & $26.47$ \\

BREA & $44.73$ & $101.39$ & $182.27$ & $294.27$ & $75.85$ & $145.30$ & $216.96$ & $287.22$\\

\rowcolor[gray]{.7} RFLPA (KRUM) & $13.60$ & $35.56$ & $46.48$ & $75.78$ & $31.77$ & $34.04$ & $39.76$ & $62.81$\\
 \bottomrule
\end{tabular}
% \end{scriptsize}
\end{table*}

\subsection{Non-IID Setting}
\label{app:noniid}

\subsubsection{Heterogenous Clients}
The previous experiments were conducted under the assumption that the local data of clients are independent and identically distributed (IID). To simulate the non-IID dataset, we adopted the setting in \cite{mcmahan2017communication} by sorting the data based on their labels and dividing them into 10,000 subsets. Consequently, the local data owned by most clients consist of only one label. 

\begin{figure}[!htp]
\begin{center}
\includegraphics[width=100mm]{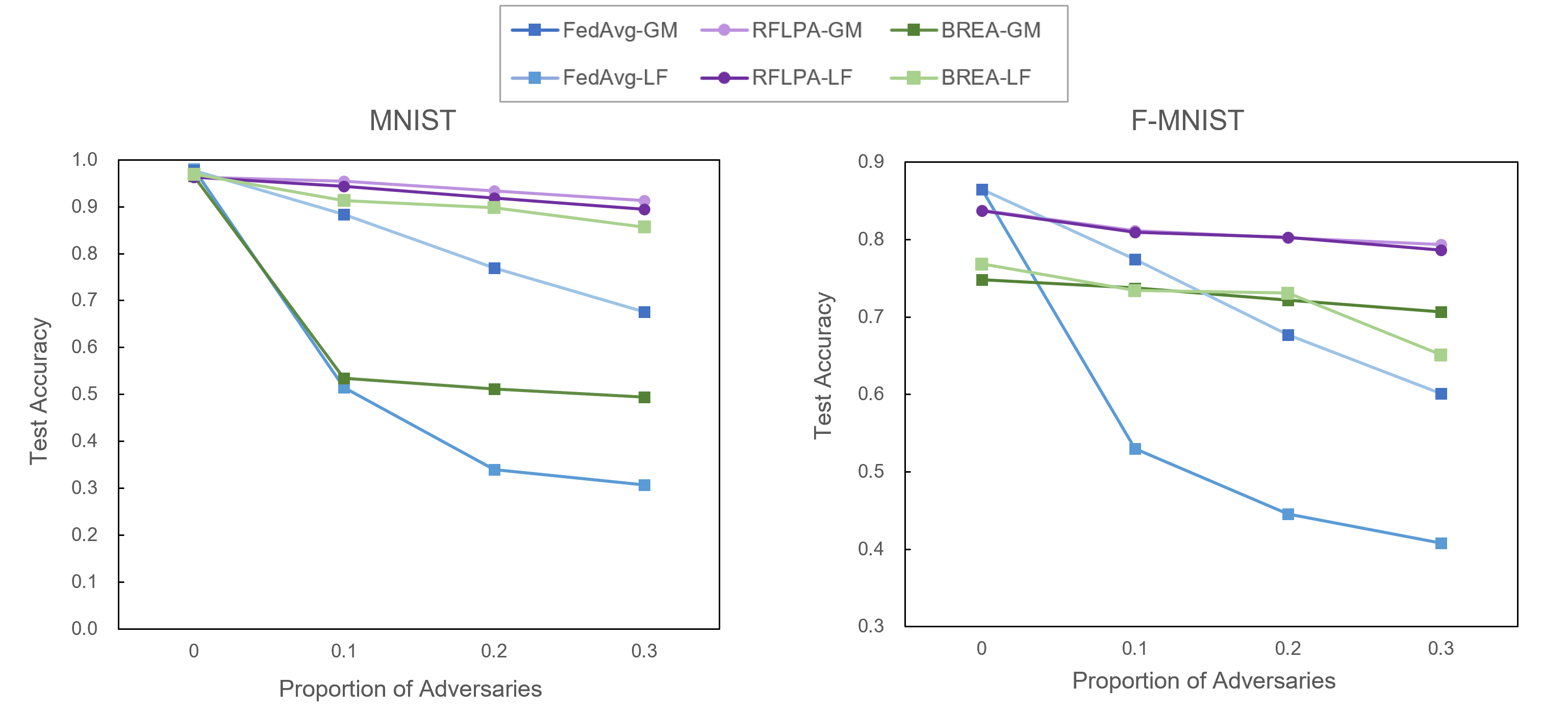}
\caption{Test accuracy on non-IID dataset. GM stands for gradient manipulation attack, and LF stands for label flipping attack.}
\label{fig:noniid}
\end{center}
\end{figure}

We compare the accuracy of RFLPA, BREA and FedAvg on non-IID dataset in Figure \ref{fig:noniid}. The RFLPA demonstrates resilient performance against poisoning attacks, even when the dataset is distributed non-identically among clients.

\subsubsection{Dynamic Data}
For dynamic settings, we consider the case where the data of the clients change during the federated training with the arrival of new data. To simulate the setting, we leverage Dirichlet Distribution Allocation (DDA) \cite{luo2021no} to sample non-iid dataset, and change the distribution for each client every 20 epochs. The parameter of the Dirichlet distribution is set to $\alpha=0.1$.

Table \ref{tab:dynamic} presents the accuracy against gradient manipulation attack. Our RFLPA demonstrates robust performance under the dynamic setting for up to 30\% attackers. The improvement of RFLPA over FedAvg is more than 2x when there are at least 20\% attackers.
\begin{table*}[!htb]
\caption{Accuracy under dynamic client data distribution against gradient manipulation attack.}
\label{tab:dynamic}
\centering
% \begin{scriptsize}
\scalebox{0.86}{
\begin{tabular}{l|cccc|cccc}
\toprule
  & \multicolumn{4}{c|}{MNIST} & \multicolumn{4}{c}{F-MNIST} \\
 \% of attackers & No & 10\% & 20\% & 30\% & No & 10\% & 20\% & 30\% \\
  \hline
FedAvg & $0.98$ \tiny$\mathbf{\pm 0.0}$\normalsize & $0.27$ \tiny$\mathbf{\pm 0.2}$\normalsize & $0.29$ \tiny$\mathbf{\pm 0.3}$\normalsize & $0.29$ \tiny$\mathbf{\pm 0.1}$\normalsize & $0.86$ \tiny$\mathbf{\pm 0.0}$\normalsize & $0.52$\tiny$\mathbf{\pm 0.1}$\normalsize & $0.21$ \tiny$\mathbf{\pm 0.2}$\normalsize & $0.18$ \tiny$\mathbf{\pm 0.2}$\normalsize \\
RFLPA & $0.96$ \tiny$\mathbf{\pm 0.0}$\normalsize & $0.94$ \tiny$\mathbf{\pm 0.0}$\normalsize & $0.92$\tiny$\mathbf{\pm 0.0}$\normalsize & $0.90$\tiny$\mathbf{\pm 0.0}$\normalsize & $0.83$\tiny$\mathbf{\pm 0.0}$\normalsize & $0.79$\tiny$\mathbf{\pm 0.0}$\normalsize & $0.79$\tiny$\mathbf{\pm 0.1}$\normalsize & $0.77$\tiny$\mathbf{\pm 0.0}$\normalsize \\ 
 \bottomrule
\end{tabular}
}
\end{table*}

\subsection{Integration with Other Aggregation Protocols}
\label{app:cleanremedy}
The robust aggregation rule of RFLPA is based on FLTrust, requiring a clean root data set on server side. Suppose we cannot get any clean root dataset even if the required size is small, it is feasible to replace the aggregation protocol with other robust aggregation algorithms to circumvent the assumption.

First, our algorithm can be integrated with KRUM-based method by substituting the aggregation module with KRUM. Though KRUM incurs greater cost than the original method, Appendix \ref{app:ablation} shows that there is a notable reduction in communication and computation cost compared with BREA, benefiting from the design of our secret sharing algorithm. The accuracy of RFLPA (KRUM) is expected to be the same as BREA, as both utilize the same aggregation rule.

Another alternative is to compute the cosine similarity with global weights. Specifically, we can compute the cosine similarity between each local update and the global weights as follows \cite{yaldiz2023secure}:
\begin{equation}
    cos(\mathbf{w}_i^t, \mathbf{w}_G^{t-1}) = \frac{\langle\mathbf{w}_i^t, \mathbf{w}_G^{t-1}\rangle}{\|\mathbf{w}_i^t\|\|\mathbf{w}_G^{t-1}\|}
\end{equation}
, and filter out the clients with similarity smaller than a pre-specified threshold, which is set to 0 in our evaluation.

From Table \ref{tab:gw}, we can observe that compared with FedAvg, RFLPA-GW effectively improves the accuracy in the presence of attackers. Noted that the communication and computation cost of RFLPA-GW is at the same scale of RFLPA’s original level, as both compute the cosine similarity with a single baseline.
\begin{table*}[!htb]
\caption{Accuracy for defense based on global weight under different proportions of attackers. RFLPA-GW replaces the robust aggregation rule in RFLPA with the method based on cosine similarity with global weight.}
\label{tab:gw}
\centering
% \begin{scriptsize}
\scalebox{0.85}{
\begin{tabular}{l|cccc|cccc}
\toprule
  & \multicolumn{4}{c|}{MNIST} & \multicolumn{4}{c}{F-MNIST} \\
 \% of attackers & No & 10\% & 20\% & 30\% & No & 10\% & 20\% & 30\% \\
  \hline
FedAvg & $\mathbf{0.98}$ \tiny$\mathbf{\pm 0.0}$\normalsize  & $0.46$ \tiny$\pm 0.1$\normalsize & $0.40$ \tiny$\pm 0.1$\normalsize & $0.32$ \tiny$\pm 0.0$\normalsize & $\mathbf{0.88}$ \tiny$\mathbf{\pm 0.0}$\normalsize  & $0.55$ \tiny$\pm 0.0$\normalsize & $0.51$ \tiny$\pm 0.0$\normalsize & $0.45$ \tiny$\pm 0.1$\normalsize \\
RFLPA-GW & $0.98$ \tiny$\mathbf{\pm 0.0}$\normalsize & $0.95$ \tiny$\mathbf{\pm 0.1}$\normalsize & $0.92$\tiny$\mathbf{\pm 0.0}$\normalsize & $0.91$\tiny$\mathbf{\pm 0.1}$\normalsize & $0.90$\tiny$\mathbf{\pm 0.0}$\normalsize & $0.80$\tiny$\mathbf{\pm 0.1}$\normalsize & $0.77$\tiny$\mathbf{\pm 0.0}$\normalsize & $0.75$\tiny$\mathbf{\pm 0.0}$\normalsize \\ 
 \bottomrule
\end{tabular}
}
\end{table*}

\section{Impact Statement}
\label{app:impact}
Our work in developing a robust federated learning framework (RFLPA) addresses significant challenges in privacy and security in federated learning (FL), presenting substantial benefits in data protection and carrying broader societal implications. The advancements in safeguarding data privacy bolster ethical standards in data handling, yet they may raise concerns in scenarios requiring data transparency. Our efforts contribute to the technical evolution of FL but also underscore the need for ongoing ethical considerations in the face of rapidly advancing machine learning technologies.

\newpage
\section*{NeurIPS Paper Checklist}

\begin{enumerate}

\item {\bf Claims}
    \item[] Question: Do the main claims made in the abstract and introduction accurately reflect the paper's contributions and scope?
    \item[] Answer: \answerYes{} % Replace by \answerYes{}, \answerNo{}, or \answerNA{}.
    \item[] Justification: 
    % \item[] Guidelines:
    % \begin{itemize}
    %     \item The answer NA means that the abstract and introduction do not include the claims made in the paper.
    %     \item The abstract and/or introduction should clearly state the claims made, including the contributions made in the paper and important assumptions and limitations. A No or NA answer to this question will not be perceived well by the reviewers. 
    %     \item The claims made should match theoretical and experimental results, and reflect how much the results can be expected to generalize to other settings. 
    %     \item It is fine to include aspirational goals as motivation as long as it is clear that these goals are not attained by the paper. 
    % \end{itemize}

\item {\bf Limitations}
    \item[] Question: Does the paper discuss the limitations of the work performed by the authors?
    \item[] Answer: \answerYes{} % Replace by \answerYes{}, \answerNo{}, or \answerNA{}.
    \item[] Justification: See Section \ref{app:limitation} Discussion and Future Work.

\item {\bf Theory Assumptions and Proofs}
    \item[] Question: For each theoretical result, does the paper provide the full set of assumptions and a complete (and correct) proof?
    \item[] Answer: \answerYes{} % Replace by \answerYes{}, \answerNo{}, or \answerNA{}.
    \item[] Justification: 
    % \item[] Guidelines:
    % \begin{itemize}
    %     \item The answer NA means that the paper does not include theoretical results. 
    %     \item All the theorems, formulas, and proofs in the paper should be numbered and cross-referenced.
    %     \item All assumptions should be clearly stated or referenced in the statement of any theorems.
    %     \item The proofs can either appear in the main paper or the supplemental material, but if they appear in the supplemental material, the authors are encouraged to provide a short proof sketch to provide intuition. 
    %     \item Inversely, any informal proof provided in the core of the paper should be complemented by formal proofs provided in appendix or supplemental material.
    %     \item Theorems and Lemmas that the proof relies upon should be properly referenced. 
    % \end{itemize}

    \item {\bf Experimental Result Reproducibility}
    \item[] Question: Does the paper fully disclose all the information needed to reproduce the main experimental results of the paper to the extent that it affects the main claims and/or conclusions of the paper (regardless of whether the code and data are provided or not)?
    \item[] Answer: \answerYes{} % Replace by \answerYes{}, \answerNo{}, or \answerNA{}.
    \item[] Justification: 

\item {\bf Open access to data and code}
    \item[] Question: Does the paper provide open access to the data and code, with sufficient instructions to faithfully reproduce the main experimental results, as described in supplemental material?
    \item[] Answer: \answerYes{} % Replace by \answerYes{}, \answerNo{}, or \answerNA{}.
    \item[] Justification: 
    % \item[] Guidelines:
    % \begin{itemize}
    %     \item The answer NA means that paper does not include experiments requiring code.
    %     \item Please see the NeurIPS code and data submission guidelines (\url{https://nips.cc/public/guides/CodeSubmissionPolicy}) for more details.
    %     \item While we encourage the release of code and data, we understand that this might not be possible, so “No” is an acceptable answer. Papers cannot be rejected simply for not including code, unless this is central to the contribution (e.g., for a new open-source benchmark).
    %     \item The instructions should contain the exact command and environment needed to run to reproduce the results. See the NeurIPS code and data submission guidelines (\url{https://nips.cc/public/guides/CodeSubmissionPolicy}) for more details.
    %     \item The authors should provide instructions on data access and preparation, including how to access the raw data, preprocessed data, intermediate data, and generated data, etc.
    %     \item The authors should provide scripts to reproduce all experimental results for the new proposed method and baselines. If only a subset of experiments are reproducible, they should state which ones are omitted from the script and why.
    %     \item At submission time, to preserve anonymity, the authors should release anonymized versions (if applicable).
    %     \item Providing as much information as possible in supplemental material (appended to the paper) is recommended, but including URLs to data and code is permitted.
    % \end{itemize}

\item {\bf Experimental Setting/Details}
    \item[] Question: Does the paper specify all the training and test details (e.g., data splits, hyperparameters, how they were chosen, type of optimizer, etc.) necessary to understand the results?
    \item[] Answer: \answerYes{} % Replace by \answerYes{}, \answerNo{}, or \answerNA{}.
    \item[] Justification: 
    % \item[] Guidelines:
    % \begin{itemize}
    %     \item The answer NA means that the paper does not include experiments.
    %     \item The experimental setting should be presented in the core of the paper to a level of detail that is necessary to appreciate the results and make sense of them.
    %     \item The full details can be provided either with the code, in appendix, or as supplemental material.
    % \end{itemize}

\item {\bf Experiment Statistical Significance}
    \item[] Question: Does the paper report error bars suitably and correctly defined or other appropriate information about the statistical significance of the experiments?
    \item[] Answer: \answerYes{} % Replace by \answerYes{}, \answerNo{}, or \answerNA{}.
    \item[] Justification: 
    % \item[] Guidelines:
    % \begin{itemize}
    %     \item The answer NA means that the paper does not include experiments.
    %     \item The authors should answer "Yes" if the results are accompanied by error bars, confidence intervals, or statistical significance tests, at least for the experiments that support the main claims of the paper.
    %     \item The factors of variability that the error bars are capturing should be clearly stated (for example, train/test split, initialization, random drawing of some parameter, or overall run with given experimental conditions).
    %     \item The method for calculating the error bars should be explained (closed form formula, call to a library function, bootstrap, etc.)
    %     \item The assumptions made should be given (e.g., Normally distributed errors).
    %     \item It should be clear whether the error bar is the standard deviation or the standard error of the mean.
    %     \item It is OK to report 1-sigma error bars, but one should state it. The authors should preferably report a 2-sigma error bar than state that they have a 96\% CI, if the hypothesis of Normality of errors is not verified.
    %     \item For asymmetric distributions, the authors should be careful not to show in tables or figures symmetric error bars that would yield results that are out of range (e.g. negative error rates).
    %     \item If error bars are reported in tables or plots, The authors should explain in the text how they were calculated and reference the corresponding figures or tables in the text.
    % \end{itemize}

\item {\bf Experiments Compute Resources}
    \item[] Question: For each experiment, does the paper provide sufficient information on the computer resources (type of compute workers, memory, time of execution) needed to reproduce the experiments?
    \item[] Answer: \answerYes{} % Replace by \answerYes{}, \answerNo{}, or \answerNA{}.
    \item[] Justification: 
    % \item[] Guidelines:
    % \begin{itemize}
    %     \item The answer NA means that the paper does not include experiments.
    %     \item The paper should indicate the type of compute workers CPU or GPU, internal cluster, or cloud provider, including relevant memory and storage.
    %     \item The paper should provide the amount of compute required for each of the individual experimental runs as well as estimate the total compute. 
    %     \item The paper should disclose whether the full research project required more compute than the experiments reported in the paper (e.g., preliminary or failed experiments that didn't make it into the paper). 
    % \end{itemize}
    
\item {\bf Code Of Ethics}
    \item[] Question: Does the research conducted in the paper conform, in every respect, with the NeurIPS Code of Ethics \url{https://neurips.cc/public/EthicsGuidelines}?
    \item[] Answer: \answerYes{} % Replace by \answerYes{}, \answerNo{}, or \answerNA{}.
    \item[] Justification: 
    % \item[] Guidelines:
    % \begin{itemize}
    %     \item The answer NA means that the authors have not reviewed the NeurIPS Code of Ethics.
    %     \item If the authors answer No, they should explain the special circumstances that require a deviation from the Code of Ethics.
    %     \item The authors should make sure to preserve anonymity (e.g., if there is a special consideration due to laws or regulations in their jurisdiction).
    % \end{itemize}

\item {\bf Broader Impacts}
    \item[] Question: Does the paper discuss both potential positive societal impacts and negative societal impacts of the work performed?
    \item[] Answer: \answerYes{} % Replace by \answerYes{}, \answerNo{}, or \answerNA{}.
    \item[] Justification: See Appendix \ref{app:impact}.

\item {\bf Safeguards}
    \item[] Question: Does the paper describe safeguards that have been put in place for responsible release of data or models that have a high risk for misuse (e.g., pretrained language models, image generators, or scraped datasets)?
    \item[] Answer: \answerNA{} % Replace by \answerYes{}, \answerNo{}, or \answerNA{}.
    \item[] Justification: 
    % \item[] Guidelines:
    % \begin{itemize}
    %     \item The answer NA means that the paper poses no such risks.
    %     \item Released models that have a high risk for misuse or dual-use should be released with necessary safeguards to allow for controlled use of the model, for example by requiring that users adhere to usage guidelines or restrictions to access the model or implementing safety filters. 
    %     \item Datasets that have been scraped from the Internet could pose safety risks. The authors should describe how they avoided releasing unsafe images.
    %     \item We recognize that providing effective safeguards is challenging, and many papers do not require this, but we encourage authors to take this into account and make a best faith effort.
    % \end{itemize}

\item {\bf Licenses for existing assets}
    \item[] Question: Are the creators or original owners of assets (e.g., code, data, models), used in the paper, properly credited and are the license and terms of use explicitly mentioned and properly respected?
    \item[] Answer: \answerYes{} % Replace by \answerYes{}, \answerNo{}, or \answerNA{}.
    \item[] Justification: 
    % \item[] Guidelines:
    % \begin{itemize}
    %     \item The answer NA means that the paper does not use existing assets.
    %     \item The authors should cite the original paper that produced the code package or dataset.
    %     \item The authors should state which version of the asset is used and, if possible, include a URL.
    %     \item The name of the license (e.g., CC-BY 4.0) should be included for each asset.
    %     \item For scraped data from a particular source (e.g., website), the copyright and terms of service of that source should be provided.
    %     \item If assets are released, the license, copyright information, and terms of use in the package should be provided. For popular datasets, \url{paperswithcode.com/datasets} has curated licenses for some datasets. Their licensing guide can help determine the license of a dataset.
    %     \item For existing datasets that are re-packaged, both the original license and the license of the derived asset (if it has changed) should be provided.
    %     \item If this information is not available online, the authors are encouraged to reach out to the asset's creators.
    % \end{itemize}

\item {\bf New Assets}
    \item[] Question: Are new assets introduced in the paper well documented and is the documentation provided alongside the assets?
    \item[] Answer: \answerNA{} % Replace by \answerYes{}, \answerNo{}, or \answerNA{}.
    \item[] Justification: 
    % \item[] Guidelines:
    % \begin{itemize}
    %     \item The answer NA means that the paper does not release new assets.
    %     \item Researchers should communicate the details of the dataset/code/model as part of their submissions via structured templates. This includes details about training, license, limitations, etc. 
    %     \item The paper should discuss whether and how consent was obtained from people whose asset is used.
    %     \item At submission time, remember to anonymize your assets (if applicable). You can either create an anonymized URL or include an anonymized zip file.
    % \end{itemize}

\item {\bf Crowdsourcing and Research with Human Subjects}
    \item[] Question: For crowdsourcing experiments and research with human subjects, does the paper include the full text of instructions given to participants and screenshots, if applicable, as well as details about compensation (if any)? 
    \item[] Answer: \answerNA{} % Replace by \answerYes{}, \answerNo{}, or \answerNA{}.
    \item[] Justification: 
    % \item[] Guidelines:
    % \begin{itemize}
    %     \item The answer NA means that the paper does not involve crowdsourcing nor research with human subjects.
    %     \item Including this information in the supplemental material is fine, but if the main contribution of the paper involves human subjects, then as much detail as possible should be included in the main paper. 
    %     \item According to the NeurIPS Code of Ethics, workers involved in data collection, curation, or other labor should be paid at least the minimum wage in the country of the data collector. 
    % \end{itemize}

\item {\bf Institutional Review Board (IRB) Approvals or Equivalent for Research with Human Subjects}
    \item[] Question: Does the paper describe potential risks incurred by study participants, whether such risks were disclosed to the subjects, and whether Institutional Review Board (IRB) approvals (or an equivalent approval/review based on the requirements of your country or institution) were obtained?
    \item[] Answer:\answerNA{} % Replace by \answerYes{}, \answerNo{}, or \answerNA{}.
    \item[] Justification: 
    % \item[] Guidelines:
    % \begin{itemize}
    %     \item The answer NA means that the paper does not involve crowdsourcing nor research with human subjects.
    %     \item Depending on the country in which research is conducted, IRB approval (or equivalent) may be required for any human subjects research. If you obtained IRB approval, you should clearly state this in the paper. 
    %     \item We recognize that the procedures for this may vary significantly between institutions and locations, and we expect authors to adhere to the NeurIPS Code of Ethics and the guidelines for their institution. 
    %     \item For initial submissions, do not include any information that would break anonymity (if applicable), such as the institution conducting the review.
    % \end{itemize}

\end{enumerate}

\end{document}